\documentclass{article}
\usepackage{endfloat}
\usepackage{graphicx}
\newenvironment*{linenomath}{}{}
\usepackage{hyperref}
\usepackage{nameref}
\usepackage{cite}
\usepackage{floatrow}
\usepackage{graphicx}
\usepackage[label font=bf]{subfig}
\usepackage{caption}
\usepackage{pdfpages}
\usepackage{multicol}        % used for the two-column index
\usepackage[bottom]{footmisc}% places footnotes at page bottom
\usepackage{subfig}
\usepackage{amsfonts}
\usepackage{hyperref}
\usepackage[cmex10]{amsmath}
\usepackage{amsmath}
\usepackage{amsthm}
\usepackage{xcolor}
\usepackage{mathtools}

\newfloatcommand{capbtabbox}{table}[][\FBwidth]
\newtheorem{thm}{Theorem}

\theoremstyle{remark}
\newtheorem{remark}{Remark}

\usepackage{bm}

\usepackage[normalem]{ulem}
\usepackage{xcolor}

\usepackage[normalem]{ulem} % to use \sout
\usepackage{xcolor}
\definecolor{forestgreen}{rgb}{0.33,0.61,0.34}
\newcommand{\add}[1]{\textcolor{black}{#1}}

\newcommand{\del}[1]{\textcolor{red}{}}
\usepackage{geometry}

\usepackage{blindtext}

\begin{document}
	%\linenumbers
	\floatsetup[figure]{style=plain,subcapbesideposition=top}
	\title{Fixation probability in evolutionary dynamics on switching temporal networks}
	\author{Jnanajyoti Bhaumik$^1$ \and Naoki Masuda$^{1,2}$}
	\date{%
		$^1$Department of Mathematics, State University of New York at Buffalo, NY 14260-2900, USA\\%
		$^2$Computational and Data-Enabled Science and Engineering Program,
		State University of New York at Buffalo, Buffalo, NY 14260-5030, USA\\[2ex]%
	}
	
	\maketitle

	\begin{abstract}
	Population structure has been known to substantially affect evolutionary dynamics. Networks that promote the spreading of fitter mutants are called amplifiers of \del{natural} selection, and those that suppress the spreading of fitter mutants are called suppressors \add{of selection}. Research in the past two decades has found various families of amplifiers while suppressors still remain somewhat elusive. It has also been discovered that most networks are amplifiers \add{of selection} under the birth-death updating combined with uniform initialization, which is a standard condition assumed widely in the literature. In the present study, we extend the birth-death processes to temporal (i.e., time-varying) networks. For the sake of tractability, we restrict ourselves to switching temporal networks, in which the network structure \add{deterministically} alternates between two static networks at constant time intervals or stochastically in a Markovian manner. We show that, in a majority of cases, switching networks are less amplifying than both of the two static networks constituting the switching networks. Furthermore, most small switching networks\add{, i.e., networks on six nodes or less,} are suppressors, which contrasts to the case of static networks.
	\end{abstract}
	
	\section{Introduction}
	
	Evolutionary dynamics models enable us to study how populations change over time under natural selection and neutral random drift among other factors.
	Over the past two decades, the population structure, particularly those represented by networks (i.e., graphs), has been shown to significantly alter the spread of mutant types~\cite{lieberman2005evolutionary,nowak_book,nowak2010evolutionary,shakarian2012review,perc2013evolutionary}.
	Mutants may have a fitness that is different from the fitness of a resident type, which makes the mutants either more or less likely to produce offsprings. The fitness of each type may vary depending on the type of the neighboring individuals' types as in the case of evolutionary games on networks. On the other hand, the simplest assumption on the fitness is to assume that the fitness of each type is constant over time. This latter case, which we refer to as constant selection,
has also been studied as biased voter models, modeling stochastic opinion formation in networks (and well-mixed populations)\cite{durrett1999stochastic,antal2006evolutionary,sood2008voter,castellano2009statistical}.
	
Networks on which real-world dynamical processes approximated by evolutionary dynamics occur may be time-varying.
Temporal (i.e., time-varying) networks and dynamical processes on them have been extensively studied~\cite{holme2012temporal,holme2013temporal,holme2015modern,masuda2017introduction,karsai2018bursty,holme2019temporal,lambiotte2021guide}. Evolutionary game dynamics on time-varying networks are no exception.
It has been shown that temporal networks enhance the evolution of cooperation as compared to static networks~\cite{cardillo2014evolutionary,li2020evolution,johnson2021temporal,sheng2021evolutionary,su2023strategy}. It has also been known for a longer time that 
coevolutionary dynamics of a social dilemma game and network structure, in which the dynamics of the network structure depend on the state of the nodes (e.g., cooperator or defector), enhance overall cooperation if players tend to avoid creating or maintaining edges connecting to defectors~\cite{santos2006cooperation,pacheco2006coevolution,fu2009partner,perc2013evolutionary,mcavoy2020social}.
	
In this study, we investigate constant-selection evolutionary dynamics on temporal networks to clarify how the time dependence of the network structure impacts evolutionary processes. In particular, a key question in studies of constant-selection evolutionary dynamics on networks is the fixation probability, defined as the probability that a single mutant type introduced to a node in the network eventually fixates, i.e., occupies all the nodes of the network. The fixation probability depends on the fitness of the mutant type relative to the fitness of the resident type, denoted by $r$.
A network is called an amplifier of \del{natural} selection if it has a higher fixation probability than the complete graph, which corresponds to the Moran process, when $r>1$ and a lower fixation probability when $r<1$; conversely, a network is called a suppressor \add{of selection} if the fixation probability is smaller than for the Moran process on $r>1$ and larger for $r<1$ \cite{lieberman2005evolutionary,adlam2015amplifiers}. In Fig.~\ref{fig:amplifier_moran_suppressor}, we show hypothetical examples of the fixation probability as a function of $r$ for three networks: the complete graph (i.e., Moran process), an amplifier, and a suppressor.
Under the so-called birth-death updating rule and uniform initialization, most static networks are amplifiers \add{of selection} \cite{hindersin2015most,allen2021fixation}. In fact, there is only one suppressing static network with six nodes among the 112 connected six-node networks \cite{alcalde2017suppressors}.

\begin{figure}[t] % [!htbp]%
	\begin{center}
		\centering
		{\includegraphics[width=7cm]{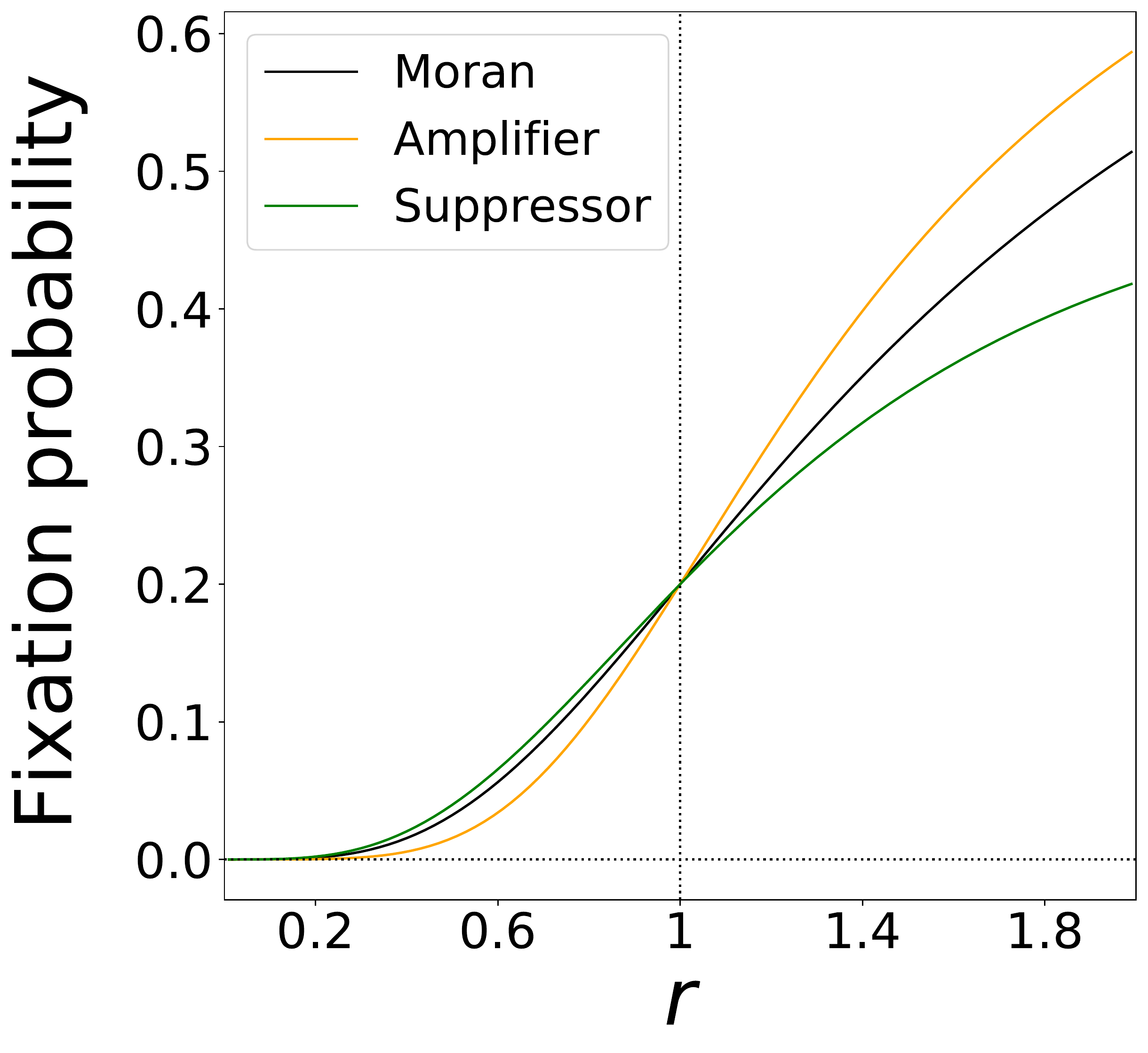} }
		\caption{Concept of amplifier and suppressor of \del{natural} selection. \add{The fitness of the resident and mutant type is equal to $1$ and $r$, respectively.
		A value of $r>1$ indicates that a mutant node is more likely to be selected for reproduction than a resident node. A value of $r<1$ indicates vice versa.}   The fixation probability of a single mutant type for an amplifier is smaller than that for the Moran process when $r<1$ and larger than that for the Moran process when $r>1$. Conversely, the fixation probability for a suppressor is larger than that for the Moran process when $r<1$ and smaller than that for the Moran process when $r>1$. The Moran process, amplifier, and suppressor have the same fixation probability at $r=1$, which is equal to $1/N$. In the figure, the fixation probabilities for the Moran process are given by Eq.~\eqref{eq:fixation-prob-Moran} with $N=5$, and those for
the amplifier and suppressor are hypothetical ones for expository purposes.
%
% \jb{No, they are fictive plots meant to represent amplifiers and suppressors, the functions I used are $\frac{1-\frac{1}{r}}{1-\frac{1}{r^5}}\cdot r^{0.3}$ for the amplifier, $\frac{1-\frac{1}{r}}{1-\frac{1}{r^5}}\cdot \frac{1}{r^{0.3}}$ for the suppressor and $\frac{1-\frac{1}{r}}{1-\frac{1}{r^5}}$ for Moran. }
%
	}
		\label{fig:amplifier_moran_suppressor}%
	\end{center}
\end{figure}

Furthermore, various families of amplifiers \add{of selection} have been found~\cite{giakkoupis2016amplifiers,galanis2017amplifiers,pavlogiannis2017amplification,pavlogiannis2018construction,goldberg2019asymptotically}, whereas suppressors \add{of selection} still remain elusive~\cite{alcalde2017suppressors,alcalde2018evolutionary}.
	On these grounds, we ask the following two main questions in the present study. First, as in the case of static networks, \del{is a vast majority} \add{are many} \del{of} temporal networks \add{comprised of sequences of unweighted networks} amplifier\add{s} of \del{natural} selection under the same condition (i.e., birth-death updating rule and uniform initialization)? Second, if we combine amplifying static networks, $G_1$ and $G_2$, into a temporal network, can the obtained temporal network be a suppressor \add{of selection} or a less amplifying temporal network than both $G_1$ and $G_2$?

	\section{Model}
	
	Let $G$ be a static weighted network with $N$ nodes. We assume undirected networks for simplicity although extending the following evolutionary dynamics to the case of directed networks is straightforward. We assume that each node takes either the resident or mutant type at any discrete time.
	The resident and mutant have fitness $1$ and $r$, respectively. The fitness represents the propensity with which each type is selected for reproduction in each time step. The mutant type initially occupies just one node, which is selected uniformly at random among the $N$ nodes. The other $N-1$ nodes are occupied by the resident type.
	We then run the birth-death process, which is a generalization of the Moran process to networks \cite{lieberman2005evolutionary, ohtsuki2006simple, olfati2007evolutionary, nowak2010evolutionary,shakarian2012review,perc2013evolutionary}. Specifically, in every discrete time step, we select a node $v$ to reproduce with the probability proportional to its fitness value. Next, we select a neighbor of $v$, denoted by $v'$, with the probability proportional to the weight of the undirected edge ($v$, $v'$). Then, the type at $v$ (i.e., either resident or mutant) replaces that at $v'$. We repeat this process until the entire population is of a single type, either resident or mutant, which we call the fixation.
	
	In this study, we extend this birth-death process to temporal networks in which two static networks $G_{1}$ and $G_{2}$, both having $N$ nodes, alternate with constant intervals $\tau$. We call this temporal network model the switching network and denote it by $(G_1, G_2, \tau)$. 
	Switching networks have been used for studying various dynamics on temporal networks including synchronization \cite{porfiri2006random,stilwell2006sufficient,olfati2007evolutionary,Naoki_klemm_eguiluz_2013,hasler2013dynamicspartone,hasler2013dynamicsparttwo,masuda2016accelerating}, random walk \cite{perra2012random,naokirocha2014random,alessandretti2017random}, epidemic processing \cite{speidel2016temporal, onaga2017concurrency,valdano2015analytical,somers2023sparse}, network control \cite{li2017fundamental}, and reaction-diffusion systems \cite{petit2017theory}.
	Specifically, we first run the birth-death process on
	$G_{1}$ for $\tau$ time steps. Then, we switch to $G_{2}$ and run the same birth-death process on $G_2$ for $\tau$ time steps.
	Then, we switch back to $G_{1}.$ We keep flipping between
	$G_{1}$ and $G_{2}$ every $\tau$ time steps until the fixation of either type occurs.

\section{Computation and theoretical properties of the fixation probability in switching networks}
	
In this section, we describe the methods for calculating the fixation probability of a single mutant, i.e., the probability that the mutant type of fitness $r$ fixates when there is initially just one node of the mutant type that is selected uniformly at random. We extend the methods for static networks \cite{hindersin2016exact} to our model. We also state some mathematical properties of the fixation probability in switching networks.
	
	\subsection{Fixation probability in static networks}
	
	We first explain the known procedure for calculating the fixation probability of the mutant type, which we simply refer to as the fixation probability in the following text, in any static weighted network using Markov chains~\cite{lieberman2005evolutionary,hindersin2016exact}. We describe the state of the evolutionary dynamics by an $N$-dimensional binary vector $\bm{s} = (s_1, \ldots, s_N)$, where $s_i \in \{0, 1\}, \forall i \in \{1, \ldots, N\}$.
	%
	%\nm{ bold means vector, nonbold means vector's entries.}
	%
	For each $i$, let $s_i = 0$ or $s_i = 1$ indicate that node $i$ is occupied by a resident or a mutant, respectively. Let $S$ be the set of all states. Note that $S$ has cardinality $2^{N}$, that is, there are $2^{N}$ states and that there are $\binom{N}{m}$ states with $m$ mutants. We label the states by a bijective map, denoted by $f$, from $S$ to $\{1,\dots,2^N\}$. The transition probability matrix of the Markov chain, denoted by $T = (T_{ij})$, is a $2^N\times 2^N$ matrix. Its entry $T_{f(\bm{s}),f(\bm{s'})}$ represents the probability that the state changes from $\bm{s}$ to $\bm{s'}$ in one time step. It should be noted that $T_{f(\bm{s}),f(\bm{s'})}$
	can be non-zero if and only if vectors $\bm{s}$ and $\bm{s'}$ differ in at most one entry. Therefore, each row of $T$ has at
	most $N+1$ non-zero entries.
	
	Let $\bm{s}$ be a state with $m$ mutants, $s_i = 1$ for $i \in \{g(1), \ldots, g(m)\}$, and $s_i = 0$ for $i \in \{g(m+1), \ldots, g(N)\}$, where $g$ is a permutation on $\{1, \ldots, N \}$. Let $\bm{s'}$
	be the state with $m+1$ mutants in which $s'_i = 1$ for $i \in \{g(1), \ldots, g(m), g(m+1)\}$ and $s'_i = 0$ for $i\in \{g(m+2), \ldots, g(N)\}$. Note that $\bm{s}$ and $\bm{s'}$ differ only at the $g(m+1)$th node, where $\bm{s}$ has a resident and $\bm{s'}$ has a mutant. We obtain
		\begin{linenomath}
	\begin{equation}
		T_{f(\bm{s}),f(\bm{s'})} = \frac{r}{rm + N-m}\sum_{m'=1}^{m}\frac{ A_{g(m'), g(m+1)} } {w(g(m'))},
	\end{equation}
	\end{linenomath}
	where $A$ denotes the weighted adjacency matrix of the network, i.e., $A_{ij}$ is the weight of edge $(i, j)$,  and $w(i) \equiv
	\sum_{j=1}^N A_{ij}$ represents the weighted degree of the $i$th node, also called the strength of the node. 
	Next, consider a state $\bm{s''}$ with $m-1$ mutants such that $s_i'' = 1$ for $i \in \{ g(1), \ldots, g(\tilde{m}-1), g(\tilde{m}+1), \ldots, g(m)\}$ and $s''_i = 0$ for 
	$i \in \{ g(\tilde{m}), g(m+1), g(m+2), \ldots, g(N) \}$, \add{ where $\tilde{m} \in \{ 1, \ldots, m \}$}. We obtain
	\begin{linenomath}
	\begin{equation}
		T_{f(\bm{s}),f(\bm{s''})}=\frac{1}{rm+N-m}\sum_{m'=m+1}^{N}\frac{A_{g(m'),g(\tilde{m})}}{w(g(m'))}.
	\end{equation}
\end{linenomath}
	The probability that the state does not change after one time step is given by
	\begin{linenomath}
	\begin{equation}
		T_{f(\bm{s}), f(\bm{s})} =
		1 - \frac{r}{rm+N-m}\sum_{\ell=m+1}^{N}\sum_{m'=1}^{m}\frac{A_{g(m'),g({\ell)}}}{w(g(m'))}
		-\frac{1}{rm+ N-m}\sum_{\tilde{m}=1}^{m}\sum_{m'={m+1}}^{N}\frac{A_{g(m'),g(\tilde{m})}}{w(g(m'))}.
	\end{equation}
	\end{linenomath}
	Let $x_{f(\bm{s})}$ denote the probability that the mutant fixates when the evolutionary dynamics start from state $\bm{s}$. Because
	\begin{linenomath}
	\begin{equation}
		x_{f(\bm{s})}=\sum_{\bm{s'}\in S}T_{f(\bm{s}),f(\bm{s'})}x_{f(\bm{s'})},
	\end{equation}
\end{linenomath}
	we obtain 
	$T\bm{x}=\bm{x}$, where $\bm{x}= \left(x_1, \dots  ,x_{2^N} \right)^{\top}$, and ${}^\top$ represents the transposition. Because $x_{f(\left(0, \ldots, 0\right))}=0$ and $x_{f(\left( 1, \ldots, 1\right)) }=1$, we need to solve the set of $2^{N}-2$ linear equations to obtain the fixation probabilities starting from an arbitrary initial state. 
	
	\subsection{Fixation probability in switching networks}
	
	We now consider the same birth-death process on switching network $(G_1, G_2, \tau)$. To calculate the fixation probability in $(G_1, G_2, \tau)$, we denote by $T^{(1)}$ and $T^{(2)}$ the transition probability matrices for the birth-death process on static network $G_1$ and $G_2$, respectively. Let $x_i(t)$ be the fixation probability when the evolutionary dynamics start from the $i$th state (with $i\in \{1, \ldots, 2^N\}$) at time $t$. We obtain
	\begin{linenomath}
	\begin{align}
		\label{eqn:switch}
		\bm{x}(t) = \begin{cases}
			T^{(1)} \bm{x}(t+1) & \text{ if } 2n\tau\le t < \left(2n+1\right)\tau, \\
			T^{(2)} \bm{x}(t+1) &\text{ if } (2n+1)\tau\le t < \left(2n+2\right)\tau,
		\end{cases}
	\end{align}
\end{linenomath}
	where $\bm{x}(t) = \left( x_1(t), \ldots, x_{2^N}(t) \right)^{\top}$ and $n\in \{0,1, \ldots \}$. We recursively use Eq.~\eqref{eqn:switch} to obtain
	\begin{linenomath}
	\begin{align}
		\bm x\left(0\right) =& T^{(1)}\bm x\left(1\right)=\cdots=\left(T^{(1)}\right)^{\tau}\bm x \left(\tau\right)
		=\left(T^{(1)}\right)^{\tau} \left(T^{(2)}\right) \bm x\left(\tau+1\right)=\cdots \notag\\
		=& \left( T^{(1)} \right)^{\tau} \left( T^{(2)} \right)^{\tau} \bm x\left(2\tau\right).
	\end{align}
\end{linenomath}
	Because of the periodicity of the switching network, we obtain $\bm x\left(0\right) = \bm x\left(2\tau\right)$. Therefore, the 
	fixation probability is given as the solution of
	\begin{linenomath}
	\begin{equation}
		\bm{x^*} = \left(T^{(1)}\right)^{\tau} \left(T^{(2)}\right)^{\tau} \bm{x^*}.
	\end{equation}
	\end{linenomath}
	Let $\tilde{S}^{(1)}$ be the set of the $N$ states with just one mutant. Then, the fixation probability when there is initially a single mutant located on a node that is selected uniformly at random
	%
	% , which we often examine \cite{lieberman2005evolutionary,hindersin2016exact,adlam2015amplifiers,masuda2009evolutionary},
	%
	is given by
	\begin{linenomath}
	\begin{equation}
		\label{eqn:fixn_prob}
		\rho \equiv \frac{1}{N} \sum_{\bm{s}\in \tilde{S}^{(1)}} x^*_{f(\bm{s})}.
	\end{equation}
\end{linenomath}
	Note that $\rho$ is a function of $r$ and depends on the network structure. Because $\left(T^{(1)}\right)^{\tau} \left(T^{(2)}\right)^{\tau}$ is a stochastic matrix with two absorbing states, it has a unique solution \cite{Broom_unique_2008,stochastic_modelling_book}.
	
	The birth-death process on switching networks has the following property.
	\begin{thm}{(Neutral drift)}\label{thm:neutral-drift}
		If $r=1$, then $\rho=\frac{1}{N}$ for arbitrary $G_1$, $G_2$, and $\tau \in \mathbb{N}$.
	\end{thm}
	
	\begin{proof}
		We imitate the proof given in~\cite{Ruodan23}. Assume a switching network $(G_1, G_2, \tau)$ on $N$ nodes and that each node is initially occupied by a mutant of distinct type, i.e., node $i$ is occupied by a mutant of type $A_i$. We also assume that each mutant has fitness $1$. We denote the probability that mutant $A_i$ fixates by $q_i$. Note that $\sum_{i=1}^{N}q_i=1$. Now we reconsider our original evolutionary dynamics with $r=1$, in which there are only equally strong two types, i.e., resident type and mutant type, with the initial condition in which the mutant type occupies the $i$th node and the resident type occupies all the other $N-1$ nodes. Then, the fixation probability of the mutant is equal to $q_i$ because this model is equivalent to the previous model if we identify $A_i$ with the mutant type and the other $N-1$ types with the resident type. Therefore, the fixation probability for the original model with $r=1$ and the uniform initialization is given by $\sum_{i=1}^{N} q_i / N = 1/N$.
	\end{proof}
	
\begin{remark}
	\add{We acknowledge that a recent study proved a more general version of this theorem and provided extensive discussion on neutral drift~\cite{su2023strategy}.}
\end{remark}
	
	\begin{remark}
	The theorem holds true even if we switch among more than two static networks or if the switching intervals, $\tau$, deterministically change from one switching interval to another. The proof remains unchanged.
	\end{remark}

\add{\subsection{Initialization at random time }}

\add{In this section, we discuss the case in which the initial mutant arises in the switching network $\left(G_1,G_2,\tau\right)$ at a time selected uniformly at random. Without loss of generality, we assume that the initial mutant arises at time $t_0$, where $t_0\in \{0,1,\dots, 2\tau -1\}$. If the initial mutant appears at time $t \geq 2N$, then we can set $t_0 = t \mod 2\tau$.  Similar to Eq.~\eqref{eqn:switch}, we obtain}
\begin{linenomath}
	\begin{align}
		\label{eqn:switch_random_time}
\add{		\bm{x}(t_0) = \begin{cases}
			T^{(1)} \bm{x}(t_0+1) & \text{ if } 0\le t_0 < \tau, \\
			T^{(2)} \bm{x}(t_0+1) &\text{ if } \tau\le t_0 < 2\tau.
		\end{cases}
		}
	\end{align}
\end{linenomath}

\add{We use Eq.~\eqref{eqn:switch_random_time} to obtain}
\begin{linenomath}
\begin{align}
	\label{eqn:switch_random_time_T_1_first}
\add{ \bm x\left(t_0\right) = } & \add{ T^{(1)}\bm x\left(t_0+1\right)=\cdots} \notag\\
\add{=} & \add{\left(T^{(1)}\right)^{\tau-t_0}\bm x \left(\tau\right)} \notag\\
\add{	=} & \add{\left(T^{(1)}\right)^{\tau-t_0} \left(T^{(2)}\right) \bm x\left(\tau+1\right)=\cdots \notag } \\
\add{	=} & \add{ \left( T^{(1)} \right)^{\tau-t_0} \left( T^{(2)} \right)^{\tau} \bm x\left(2\tau\right)} \notag\\
\add{=} & \add{\left( T^{(1)} \right)^{\tau-t_0} \left( T^{(2)} \right)^{\tau}\left( T^{(1)} \right)^{t_0} \bm x\left(2\tau+t_0\right)
	}
\end{align}
\end{linenomath}
\add{when $0\leq t_0 <\tau$ and}

\begin{linenomath}
\begin{align}
	\label{eqn:switch_random_time_T_2_first}
\add{	\bm x\left(t_0\right) = } & \add{ T^{(2)}\bm x\left(t_0+1\right)=\cdots} \notag\\ % =\left(T^{(2)}\right)^{2\tau-t_0}\bm x \left(2\tau\right)
%	=\left(T^{(2)}\right)^{2\tau-t_0} \left(T^{(1)}\right) \bm x\left(\tau+1\right)=\cdots \notag } \\
	%
% \add{	=}& \add{ \left( T^{(2)} \right)^{2\tau-t_0} \left( T^{(1)} \right)^{\tau} \bm x\left(3\tau\right)
\add{=} & \add{\left( T^{(2)} \right)^{2\tau-t_0} \left( T^{(1)} \right)^{\tau}\left( T^{(2)} \right)^{t_0-\tau} \bm x\left(2\tau+t_0\right)}
\end{align} 
\end{linenomath}
\add{when $\tau\leq t_0 <2\tau$.}
\add{Because of the periodicity of the switching network, we obtain $\bm x\left(t_0\right) = \bm x\left(2\tau+t_0\right)$. Therefore, the 
fixation probability, which depends on $t_0$ in the present case, is given as the solution of}
	\begin{linenomath}\add{\begin{align}
	\label{eqn:switch_random_time_soln}
	\bm{x^*}(t_0)  = \begin{cases}
\left( T^{(1)} \right)^{\tau-t_0} \left( T^{(2)} \right)^{\tau}\left( T^{(1)} \right)^{t_0} \bm{x^*}(t_0) & \text{ if } 0\le t_0 < \tau, \\
		 \left( T^{(2)} \right)^{2\tau-t_0} \left( T^{(1)} \right)^{\tau}\left( T^{(2)} \right)^{t_0-\tau} \bm{x^*}(t_0) &\text{ if } \tau\le t_0 < 2\tau.
	\end{cases}
\end{align}}
\end{linenomath}
\add{Equation~\eqref{eqn:switch_random_time_soln} yields}
\begin{linenomath}
\begin{align}
	\label{eqn:switch_random_time_simplify}
	\add{\bm {y^{*}}} \equiv & \add{T^{\left(1\right)}\bm {x^*}\left(t_0\right)} \notag\\
	\add{=} & \add{ \left( T^{(1)} \right)^{\tau-t_0+1} \left( T^{(2)} \right)^{\tau}\left( T^{(1)} \right)^{t_0-1}T^{(1)} \bm{x^*}\left(t_0\right)} \notag\\
	\add{=} & \add{\left( T^{(1)} \right)^{\tau-(t_0-1)} \left( T^{(2)} \right)^{\tau}\left( T^{(1)} \right)^{t_0-1} \bm{y^*} .} \notag\\	%
\end{align}
\end{linenomath}
\add{Therefore, we obtain}
\add{$\bm {y^*} = \bm{x^*}\left(t_0 -1\right)$ when $1\leq t_0<\tau$. Using this relationship recursively, we obtain}
\begin{linenomath}
\begin{align}
	\label{eqn:switch_random_time_simplify_further_2}
	\add{\add{ \bm{x^*}\left(\tau-1-k\right)=} \left(T^{\left(1\right)}\right)^{k}\bm{x^*}\left(\tau-1\right)  }  \notag\\
\end{align}
\end{linenomath}
\add{for $0\leq k <\tau -1$. Similarly, we obtain}
\begin{linenomath}
\begin{align}
	\label{eqn:switch_random_time_simplify_further_3}
	\add{\add{ \bm{x^*}\left(2\tau-1-k\right) =} \left(T^{\left(2\right)}\right)^{k}\bm{x^*}\left(2\tau-1\right) }  \notag\\
\end{align}
\end{linenomath}
\add{for $0\leq k <\tau -1$. By combining Eqs.~\eqref{eqn:switch_random_time_simplify_further_2} and \eqref{eqn:switch_random_time_simplify_further_3}, we obtain}
\add{
	\begin{linenomath}
\begin{equation}
	\label{eqn:fixn_probab_soln_random_time_vector}
	\bm{x^*} = \frac{1}{2\tau}
	\left\{ \left[
	\sum_{k=0}^{\tau-1}\left(T^{\left(1\right)}\right)^{k}\right] \bm{x^*}\left(\tau-1\right) + \left[ \sum_{k=0}^{\tau-1}\left(T^{\left(2\right)}\right)^{k} \right]
	\bm{x^*}\left(2\tau-1\right)\right\},
\end{equation}
\end{linenomath}
}
\add{where $\bm{x^*}$ is the fixation probability vector when the initial mutant appears at a uniformly randomly drawn time.}

\add{As in Eq.~\eqref{eqn:fixn_prob}, let $\tilde{S}^{(1)}$ be the set of the $N$ states with just one mutant. Then, the fixation probability for a single mutant when the initial mutant appears at a uniformly randomly drawn time 
	%
	% , which we often examine \cite{lieberman2005evolutionary,hindersin2016exact,adlam2015amplifiers,masuda2009evolutionary},
	%
	is given by
	\begin{linenomath}
	\begin{equation}
		\label{eqn:fixn_probab_time_t_0}
		\rho = \frac{1}{N} \sum_{\bm{s}\in \tilde{S}^{(1)}} x^*_{f(\bm{s})}.
	\end{equation}
\end{linenomath}
}
%\add{Using Eq.~\eqref{eqn:fixn_probab_time_t_0}, the fixation probability when the initialization time is uniformly chosen, is given by,
	%\begin{equation}
		%\label{eqn:fixn_probab_soln_random_time}
		%\rho \equiv %\frac{1}{2\tau}\sum_{t_0=0}^{2\tau-1}\rho_{t_0} .
%\end{equation}}

\add{\subsection{Stochastic switching\label{sub:stochastic-switching}}}

	\add{In this section, we formulate the fixation probability for stochastic switching networks. We adapt the methods proposed for epidemic spreading~\cite{ogura2016stability,ogura2016epidemic} and evolutionary games \cite{su2023strategy} to the case of constant-selection dynamics. We assume that the network switches with probability $p$ at every time step. In other words, if the network is $G_{1}$ at time $t$, then
		it switches to $G_2$ at time $t+1$ with probability $p$ and remains $G_1$ with probability $1-p$. Likewise, if the network is $G_2$ at time $t$, then it switches to $G_1$ at time $t+1$ with probability $p$ and remains $G_2$ with probability $1-p$. The duration of $G_1$ and that of $G_2$ before switching to the other network, $\tau$, obeys the geometric distribution with ${\displaystyle \Pr(\tau)=(1-p)^{\tau-1}p}$, where $\Pr$ denotes the probability.}

\add{We can write the state of the dynamics at any time $t$ as $\left(\bm{s}, G_{i}\right)$, where
		$\bm{s}$ is one of the $2^{N}$ states (i.e., $\bm{s} \in S$) as in the deterministic switching
		case, and $i \in \{ 1, 2 \}$. Let $x_{\left(\bm{s}, G_{i}\right)}\left(t\right)$ denote the probability
		that the dynamics attains fixation when starting in state
		$\bm{s}$ at time $t$. We obtain}
	\begin{linenomath}
		\begin{equation}
			\label{eqn:stochastic_recursive}
			\add{x_{\left(\bm{s}, G_{i}\right)}\left(t\right)=\sum_{\bm{s'} \in S} T_{\bm{s} \rightarrow \bm{s'}}^{(i)}\left[p\cdot x_{\left(\bm{s'}, G_{i'}\right)}\left(t+1\right)+\left(1-p\right)\cdot x_{\left(\bm{s'}, G_{i}\right)}\left(t+1\right)\right],}
		\end{equation}
	\end{linenomath}
\add{where $i' = 2$ if $i=1$ and $i' = 1$ if $i = 2$. Let $\tilde{T}$  be the $\left(2^{N}\times2\right)\times\left(2^{N}\times2\right)$ transition probability matrix defined by $\tilde{T}_{\left(\left(\bm{s}, G_{i}\right),\left(\bm{s'}, G_{i'}\right)\right)}=pT_{\bm{s} \rightarrow \bm{s'}}^{(i)}$ and $\tilde{T}_{\left(\left(\bm{s}, G_{i}\right),\left(\bm{s'}, G_{i}\right)\right)}=\left(1-p\right)T_{\bm{s} \rightarrow \bm{s'}}^{(i)}$ for $i \in \{1, 2 \}$. Matrix $\tilde{T}$ is the following block matrix:}
		\add{\begin{linenomath}\begin{equation}
		\tilde{T}=\begin{bmatrix}pT^{(1)} & \left(1-p\right)T^{(1)}\\
			pT^{(2)} & \left(1-p\right)T^{(2)}
		\end{bmatrix}.
		\end{equation}\end{linenomath}}
\add{Using $\tilde{T}$, we rewrite Eq.~\eqref{eqn:stochastic_recursive} as}
%		\add{\begin{equation}
%			x_{\left(\bm{s}, G_{i}\right)}\left(t\right)=\sum_{\bm{s'} \in S, j \in \{1, 2\}} \tilde{T}_{\left(\left(\bm{s}, G_{i}\right), \left(\bm{s'}, G_{j}\right)\right)}x_{\left(\bm{s'}, G_{j}\right)}\left(t+1\right).
%		\end{equation}}
\begin{linenomath}
		\begin{equation}
			\add{\bm{x}\left(t\right)=\tilde{T}\bm{x}\left(t+1\right),}
		\end{equation}
	\end{linenomath}
where		
		\add{\begin{linenomath}\begin{equation}
		\bm{x}\left(t\right) \equiv \begin{bmatrix}x_{\left(\left(0,\ldots,0\right),G_{1}\right)}\left(t\right)\\
			\vdots\\
			x_{\left(\left(1,\ldots,1\right),G_{1}\right)}\left(t\right)\\
			x_{\left(\left(0,\ldots,0\right),G_{2}\right)}\left(t\right)\\
			\vdots\\
			x_{\left(\left(1,\ldots,1\right),G_{2}\right)}\left(t\right)
		\end{bmatrix}.
		\end{equation}\end{linenomath}}
		\add{In fact, $\bm{x}\left(t\right)$ does not depend on $t$. Therefore, to find the fixation probability, we need to solve}
		\begin{linenomath}\begin{equation}
			\add{\bm{x^*}=\tilde{T}\bm{x^*}.}
		\end{equation}\end{linenomath}
		\add{Similar to the derivation of Eq.~\eqref{eqn:fixn_prob}, we find that the fixation probability when there is initially just one mutant on a node selected uniformly at random and the initial network is selected uniformly at random is given by}
		\begin{linenomath}\begin{equation}
			\add{\rho =\frac{1}{2N} \sum_{i=1}^2 \sum_{\bm{s} \in \tilde{S}^{(1)}}x_{\left(\bm{s}, G_i\right)}.}
	\end{equation} \end{linenomath}
	
\subsection{Identifying amplifiers and suppressors \add{of selection}}
	
	We operationally define amplifiers and suppressors \add{of selection} as follows; similar definitions were used in the literature \cite{lieberman2005evolutionary,voorhees2013birth}. For a given switching or static network, we computed the fixation probability for several values of $r$. We say that the network is amplifier \add{of selection} if the fixation probability is larger than for that for the complete graph with the same number of nodes, or equivalently, the Moran process, at six values of $r>1$, i.e., $r \in \{ 1.1,1.2,1.3,1.4,1.6,1.8 \}$ and a smaller than that for the Moran process at three values of $r<1$, i.e., $r \in \{ 0.7,0.8,0.9 \}$. Note that the fixation probability for the Moran process with $N$ individuals is given by (see e.g.\,\cite{nowak_book})
	\begin{linenomath}\begin{equation}
		\rho = \frac{1-\frac{1}{r}}{1-\frac{1}{r^{N}}}.
		\label{eq:fixation-prob-Moran}
	\end{equation}\end{linenomath}
	Similarly, we say that a network is suppressor \add{of selection} if the fixation probability is smaller than for the Moran process at the same six values of $r$ larger than $1$ and larger than for the Moran process at the three values of $r$ smaller than $1$. It is known that some static networks are neither amplifier nor suppressor \add{of selection} \cite{alcalde2018evolutionary}.
	
\add{We note that the Moran process is equivalent to the switching network in which both $G_1$ and $G_2$ are the complete graph. In this manner, one can regard that the comparison between a general switching network and the Moran process is that between two temporal networks instead of that between a temporal network and a static network.}
	%
	% \nm{Also cite Sharma and Traulsen PNAS 2022 ? Please read and judge. If you think appropriate, please citei t here together with alcalde2018.}\jb{We can cite it, though their model is slightly different, it allows for new mutants (not mutant offsprings) to be introduced throughout the dynamics instead of just once at the beginning.} \nm{Good catch. Then, let's not cite it. I've uncited it.}

	\subsection{Isothermal theorem}
	
	A network is called isothermal if its fixation probability is the same as that for the Moran process, i.e., if Eq.~\eqref{eq:fixation-prob-Moran} holds true \cite{lieberman2005evolutionary}. A static undirected network, which may be weighted, is isothermal if and only if all the nodes have the same (weighted) degree \cite{lieberman2005evolutionary,allen2019evolutionary,broom_book}. One can easily construct isothermal switching networks as follows.
	
	\begin{thm}
		If $G_1$ and $G_2$ are isothermal networks, then the switching network $\left(G_1,G_2,\tau\right)$ is an isothermal network.
	\end{thm}
	
	\begin{proof}
		The proof is exactly the same as in the static network case as shown in \cite{lieberman2005evolutionary,nowak_book}. We denote by $p_{m,m-1}$ the probability that the state of the network moves from a state with $m$ mutants to a state with $m-1$ mutants in one time step. Similarly, we denote by $p_{m,m+1}$ the probability that the state moves from one with $m$ mutants to one with $m+1$ mutants in one time step.
		We observe that $p_{m,m-1}/p_{m,m+1}=1/r$ at every time step $t$ because the static network at any $t$, which is either $G_1$ or $G_2$, is isothermal. Therefore, the fixation probability for $\left(G_1, G_2, \tau\right)$ is given by Eq.~\eqref{eq:fixation-prob-Moran}.
	\end{proof}
	
\begin{remark}
\add{The theorem including the present proof holds true both when we initially use $G_1$ for time $\tau$ and when the mutant arises at a time selected uniformly at random.}
\end{remark}

	\section{Fixation probability in various switching networks}
	
	In this section, we analyze the fixation probability in three types of switching networks, i.e., networks with six nodes, larger switching networks in which $G_1$ and $G_2$ have symmetry (i.e., complete graph, star graph, and bipartite networks), and empirical networks.
	
	\subsection{Six-node networks}	
	
	We first analyzed the fixation probability in switching networks that are composed of two undirected and unweighted connected networks with $6$ nodes. There are 112 non-isomorphic undirected connected networks on 6 nodes. We switched between any ordered pair of different networks, giving us a total of $112\times 111=12432$ switching networks. It should be noted that swapping the order of $G_1$ and $G_2$ generally yields different fixation probabilities.
	We randomly permuted the node labels in $G_2$. We did not consider all possible labeling of nodes because there would be at most $112\cdot111\cdot6!=8951040$ switching networks on 6 nodes if we allow shuffling of node labeling, although the symmetry reduces this number.
	
	In Fig.~\ref{fig:ampsbecomesup}(a), we show two arbitrarily chosen static networks on six nodes, $G_1$ and $G_2$, which are amplifiers \add{of selection} as static networks. In Fig.~\ref{fig:ampsbecomesup}(b), we plot the fixation probability as a function of the fitness of the mutant, $r$, for the switching network $(G_1, G_2, \tau=1)$, the static networks $G_1$ and $G_2$, the aggregate weighted static network generated from $G_1$ and $G_2$, and the Moran process (i.e., complete graph on six nodes). The aggregated weighted static network is the superposition of $G_1$ and $G_2$ such that the weight of the edge is either 1 or 2.
	It is equivalent to the average of $G_1$ and $G_2$ over time. All these static and switching networks yield $\rho = 1/N = 1/6$ at $r=1$, as expected (see Theorem~\ref{thm:neutral-drift}). In addition, there exist differences in $\rho$ between the different networks and the Moran process although
	the difference is small. In fact, $G_1$ and $G_2$ are amplifiers \add{of selection},
	with their fixation probability being larger than that for the Moran process when $r>1$ and vice versa when $r<1$, 
	confirming the known result~\cite{hindersin2015most,alcalde2017suppressors}. Figure~\ref{fig:ampsbecomesup}(b) also indicates that the aggregate network is an amplifier \add{of selection}. However, the switching network is suppressor \add{of selection}. 
	
	We reconfirm these results in Fig.~\ref{fig:ampsbecomesup}(c), in which we show the difference in the fixation probability between a given static or switching network and the Moran process.
	If the difference is negative for $r<1$ and positive for $r>1$, then the network is an amplifier \add{of selection}.
	If the difference is positive for $r<1$ and negative for $r>1$, then the network is a suppressor \add{of selection}.
	Figure~\ref{fig:ampsbecomesup}(c) shows that $G_1$ is a stronger amplifier than $G_2$\del{, which} \add{and that $G_2$} is a stronger amplifier than the aggregate network.
	In contrast, the switching network $(G_1, G_2, 1)$ is a suppressor \add{of selection,} while $(G_1, G_2, 10)$ and $(G_1, G_2, 50)$ are amplifiers \add{of selection}.
	The result for $(G_1 ,G_2, 50)$ is close to that for static network $G_1$, which is because the evolutionary dynamics on $(G_1, G_2, \tau)$ is equivalent to that on $G_1$ in the limit $\tau\to\infty$. \add{In practice, fixation for networks on six nodes occurs within $50$ time steps in many cases, which renders $(G_1 ,G_2, 50)$ close to $G_1$. However, we have included the results for $\tau=50$ because fixation does not occur within 50 time steps in many other cases. When the number of nodes, $N$, is large, $(G_1, G_2, \tau)$ is a genuine switching network because the fixation times are typically much longer than $N$ \cite{hindersin2016exact}.} We conclude that switching networks composed of two amplifiers can be a suppressor, in particular when $\tau$ is small. We emphasize that this counterintuitive result is not due to the property of the aggregate network because the aggregate network, which is the time average of $G_1$ and $G_2$, is also an amplifier.

\add{We show the results for the switching network with the order of $G_1$ and $G_2$ reversed and those for random initialization time in \hyperref[appendix:A1]{Appendix A}. We find that both $(G_2,G_1,1)$ and the switching network with $\tau=1$ and random initialization time are suppressors of selection. The fixation probability for the switching network with $\tau=1$ and random initialization time is the average of that for $(G_1,G_2,1)$ and $(G_2,G_1,1)$. Therefore, the fixation probability for the former lies between that for $(G_1, G_2, 1)$ and $(G_2, G_1, 1)$ at each value of $r$. Switching networks $(G_2, G_1, \tau)$ and those with random initialization time are amplifiers of selection when $\tau$ is larger (i.e., $\tau \in \{ 10, 50 \}$); this result is qualitatively the same as that for $(G_1,G_2,\tau)$. } 
%
% %In Fig.~\ref{fig:ampsbecomesup}(d) and (e), we show similar plots when $\tau=1$ and the mutant arises at a time selected uniformly at random. \nm{Is this what you mean? "random initialization" is vague, so I rephrased it to try to be more explicit.}\jb{Yes, thank you.} Because the fixation probability for the switching network $\left(G_2,G_1,1\right)$ is reasonably similar to that for $\left(G_1,G_2,1\right)$ and both $\left(G_1,G_2,1\right)$ and $\left(G_2,G_1,1\right)$ are suppressors of selection, the fixation probability for the switching network with random initialization time, which is given by the average of the two values, is also a suppressor of selection.}

To investigate the generality of this finding to other six-node networks, we calculated the fixation probability for the switching networks derived from all possible pairs of six-node networks. Table~\ref{switchstats} shows the number of switching networks on six nodes that are amplifiers \add{of selection}, that of suppressors \add{of selection}, and that of networks that are neither amplifier or suppressor, for four values of $\tau$. The table indicates that a majority of the six-node switching networks investigated are suppressors \add{of selection} when $\tau = 1$ and $\tau=3$. This result is in stark contrast to the fact that there is only 1 suppressor \add{of selection} among 112 six-node static unweighted networks under the birth-death process~\cite{alcalde2017suppressors,hindersin2015most}. Out of the 111 static networks that are not suppressor, 100 networks are amplifiers, five are isothermal, and the other six networks are neither amplifier, suppressor, nor isothermal~\cite{alcalde2018evolutionary,10.1007/978-3-319-56154-7_20}. Most switching networks are amplifiers when $\tau=50$, which is presumably because most static networks are amplifiers and the birth-death process on $(G_1, G_2, \tau)$ converges to that on $G_1$ in the limit $\tau\to\infty$, as we discussed above.

\add{We also examined the fixation probability for six-node stochastic switching networks introduced in section~\ref{sub:stochastic-switching}. As in the case with deterministic switching, we considered $112\times 111$ ordered pairs of networks and permuted the node labels of $G_2$ uniformly at random. We show in~Table~\ref{table:switchstats_stochastic} the number of amplifier of selection, suppressor of selection, and neither type, assuming random initialization time, for $p \in \{0.3, 0.5, 0.8 \}$. We find that a substantial fraction of these stochastic switching networks is suppressors of selection for each of the three $p$ values (i.e., 36.5\% for $p=0.3$;  47.3\% for $p=0.5$; 24.3\% for $p=0.8$). These results suggest that the abundance of suppressing switching networks among six-node switching networks is not due to the periodic switching nature of our switching network model.}
	
	\begin{figure}[!htbp]%
		\begin{center}
			\centering
			{\includegraphics[width=13cm]{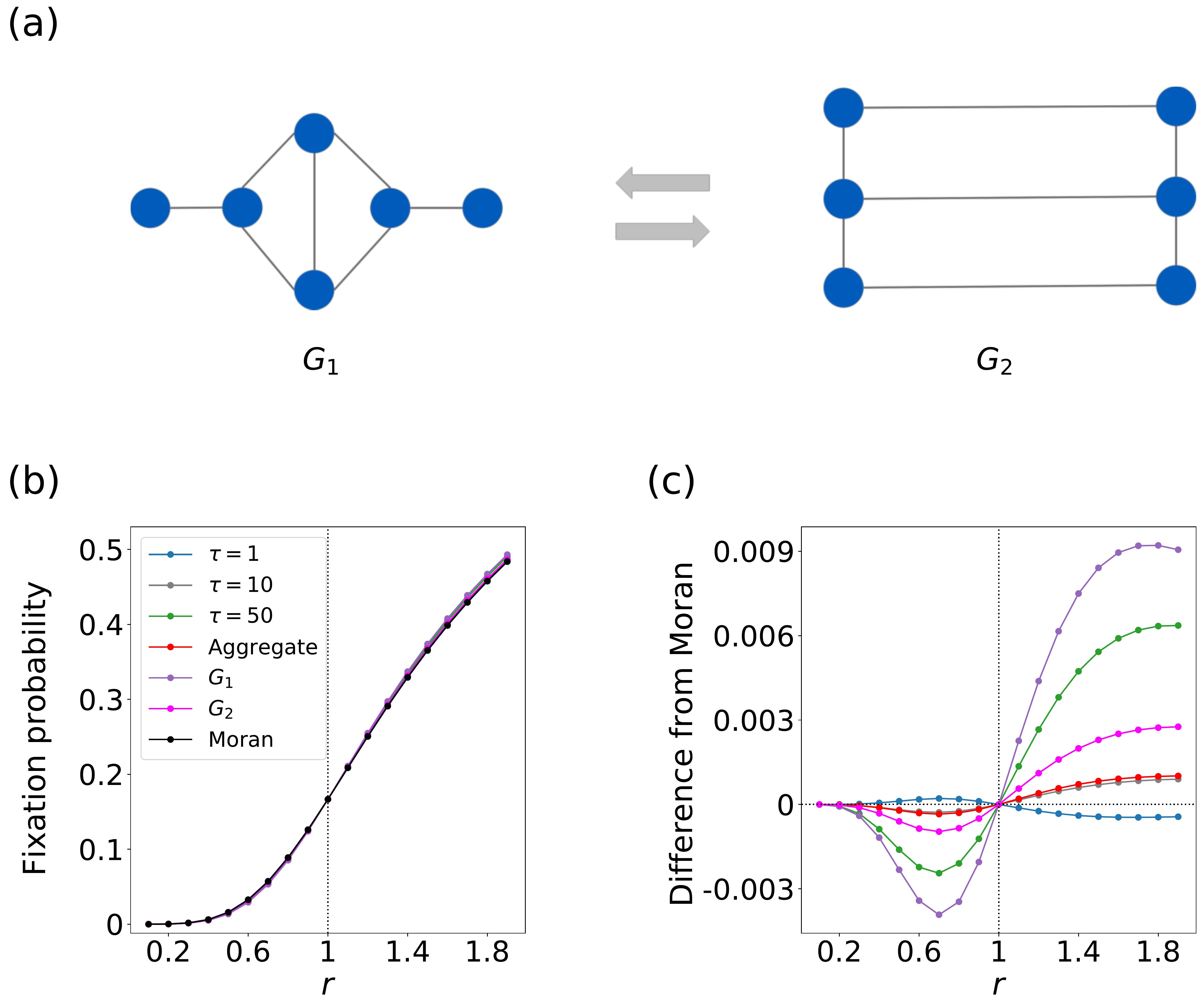} }
			\caption{A suppressing switching network composed of two amplifying static networks on six nodes. (a) A switching network composed of six nodes. Both $G_1$ and $G_2$ are amplifiers \add{of selection}. (b) Fixation probability in the static and switching networks as a function of $r$. Moran refers to the Moran process. Note that $G_1$, $G_2$, the aggregate network, and the Moran process represent static networks. (c) Difference between the fixation probability for the given network and that for the Moran process.}
			\label{fig:ampsbecomesup}%
		\end{center}
	\end{figure}

	\subsection{Larger symmetric switching networks}
	
	In this section, we assume symmetry in $G_1$ and $G_2$ to calculate the fixation probability for larger switching networks. Specifically,
	we set $G_1$ to be the star graph and $G_2$ to be either the complete graph or complete bipartite graph.
	\begin{linenomath}
	\begin{table}[t] % [!htbp]
		\centering
		\begin{tabular}{cccc} 
			\hline
			$\tau$ & Amplifier & Suppressor  & Neither\\ \hline
			1& 3636& 8177&  619\\
			3& 5190& 6347&  895\\
			10& 11102& 629& 701\\
			50& 12038& 262& 132\\ \hline
			\hline
		\end{tabular}
		{\caption{Number of amplifiers and suppressors \add{of selection} among $112\cdot 111= 12432$ \add{periodically} switching networks on six nodes.}
			\label{switchstats}}
	\end{table}
		\end{linenomath}
		\add{\begin{table}[t] % [!htbp]
		\centering
		\add{\begin{tabular}{cccc} 
			\hline
			$p$ & Amplifier & Suppressor  & Neither\\ \hline
			0.3& 7346& 4536&  550\\
			0.5& 5979& 5880&  573\\
			0.8& 8881& 3023& 528\\ \hline
			\hline
		\end{tabular}}
		{\caption{\add{Number of amplifiers and suppressors of selection among $112\cdot 111= 12432$ stochastically switching networks on six nodes.}}
			\label{table:switchstats_stochastic}}
	\end{table}}

	\subsubsection{Combination of the star graph and the complete graph\label{sub:star-complete}} 
	
	Consider switching networks in which $G_1$ is the star graph and $G_2$ is the complete graph. For this switching network, we can reduce the dimension of the transition probability matrix from $2^N \times 2^N$ to $2N\times2N$ by exploiting the symmetry in $G_1$ and $G_2$. Therefore, one can reduce the number of equations from $2^N -2$ to $2N-2$. Specifically, one can uniquely describe the state of the network by $(i, j)$, where $i\in \{0, 1\}$ and $j \in \{0, \ldots, N-1\}$. We set $i=0$ and $i=1$ when the hub node of $G_1$ is occupied by a resident or mutant, respectively.
	We set $j \in \{ 0, 1, \ldots, N-1 \}$ to the number of mutants in the other $N-1$ nodes, which we refer to as the leaves. Tuple $(i, j)$ is a valid expression of the state of the network because the $N-1$ leaves are structurally equivalent to each other in both $G_1$ and $G_2$. Tuples $(0, 0)$ and $(1, N-1)$ correspond to the fixation of the resident and mutant type, respectively.
	
	The transition probability from state $\left(i,j\right)$ to state $\left(i',j'\right)$ in a single time step of the birth-death process is nonzero if and only if $(i', j') = (i, j + 1)$ and $i=1$, $(i', j') = (i, j-1)$ and $i=0$, $(i', j') = (1-i, j)$, or $(i', j') = (i, j)$.
	Let $T^{(1)}$ denote the transition probability matrix for the star graph. We obtain
	\begin{linenomath}
	\begin{equation}
		T^{(1)}_{\left(i,j\right)\rightarrow\left(i',j'\right)} = \begin{cases}
			\frac{rj}{C_1} & \text{if }i=0 \text{ and }i'=1,
			\\[1mm]
			\frac{N-1-j}{C_2} & \text{if }i=1 \text{ and }i'=0,
			\\[1mm]
			\frac{1}{C_1}\cdot\frac{j}{N-1} & \text{if } i'=i=0 \text{ and }j' = j-1,
			\\[1mm]
			\frac{r}{C_2}\cdot\frac{N-1-j}{N-1}  & \text{if } i'=i=1 \text{ and } j' = j+1,
			\\[1mm]
			1-\sum\limits_{\mathclap{\substack{(i'',j'')\neq \\ (i,j)}}}T^{\left(1\right)}_{(i,j) \rightarrow (i'',j'')} & \text{if } (i',j')=(i, j),
			\\[1mm]
			0 & \text{ otherwise,}
		\end{cases}
		\label{eq:star-complete-T1}
	\end{equation}\end{linenomath}
	where $C_1= rj+N-j $ and $ C_2 = r(j+1)+N-\left(j+1\right) $  \cite{lieberman2005evolutionary}.
	The first line of Eq.~\eqref{eq:star-complete-T1} represents the probability that the type of the hub changes from the resident to mutant. For this event to occur, one of the $j$ leaf nodes occupied by the mutant must be chosen as parent, which occurs with probability $rj/\left(rj+N-j\right)$. Because every leaf node is only adjacent to the hub node, the hub node is always selected for death if a leaf node is selected as parent. Therefore, the probability of $i$ changing from $0$ to $1$ is equal to $rj/\left(rj+N-j\right)$, which is shown in the first line of Eq.~\eqref{eq:star-complete-T1}. As another example, consider state $\left(1,j\right)$, in which the hub has a mutant, $j$ leaf nodes have mutants, and the other $N-1-j$ leaf nodes have residents. For the state to change from $\left(1,j\right)$ to $ \left(1,j+1\right) $, the hub node must be selected as parent with probability $r/\left[r\left(j+1\right)+N-\left(j+1\right)\right]$, and a leaf node of the resident type must be selected for death, which occurs with probability $(N-1-j)/(N-1)$. The fourth line of Eq.~\eqref{eq:star-complete-T1} is equal to the product of these two probabilities.
	One can similarly derive the other lines of Eq.~\eqref{eq:star-complete-T1}.
	
	The transition probability matrix for $G_2$, which is the complete graph, is given by
	\begin{linenomath}
	\begin{equation}
		T^{(2)}_{\left(i,j\right)\rightarrow\left(i',j'\right)} = \begin{cases}
			\frac{rj}{C_1}\cdot\frac{1}{N-1}  & \text{if }i=0 \text{ and }i'=1,
			\\[1mm]
			\frac{N-1-j}{C_2}\cdot\frac{1}{N-1} & \text{if } i=1 \text{ and } i'=0,
			\\[1mm]
			\frac{N-j}{C_1}\cdot\frac{j}{N-1} & \text{if } i' = i=0 \text{ and } j' = j-1,
			\\[1mm]
			\frac{rj}{C_1}\cdot\frac{N-1-j}{N-1} & \text{if } i' = i=0 \text{ and } j' = j + 1,
			\\[1mm]
			\frac{N-1-j}{C_2}\cdot\frac{j}{N-1} & \text{if } i' = i=1 \text{ and } j' = j-1,
			\\[1mm]
			\frac{r(j+1)}{C_2}\cdot\frac{N- 1- j}{N-1}  & \text{if } i' = i=1 \text{ and } j' = j + 1,
			\\[1mm]
			1-\sum\limits_{\mathclap{\substack{(i'',j'')\neq \\ (i,j)}}}T^{\left(1\right)}_{(i,j) \rightarrow (i'',j'')} & \text{if } (i',j')=(i, j),
			\\[1mm]
			0 & \text{ otherwise.}
		\end{cases}
		\label{eq:star-complete-T2}
	\end{equation}
\end{linenomath}
	For example, for the transition from state $\left(0,j\right) $ to $ \left(1,j\right)$ to occur,  one of the $j$ mutant leaf nodes must be first selected as parent, which occurs with probability $rj/\left(rj+N-j\right)$.
	Then, the hub node must be selected for death, which occurs with probability $1/\left(N-1\right)$.
	The first line of Eq.~\eqref{eq:star-complete-T2} is equal to the product of these two probabilities.
	As another example, for the state to change from
	$ \left(1,j\right) $ to  $ \left(1,j+1\right) $, one of the mutant nodes, which may be the hub or a leaf, must be first selected as parent, which occurs with probability $ r\left(j+1\right)/\left[r\left(j+1\right)+N-\left(j+1\right)\right]$. Then, a leaf node of the resident type must be selected for death, which occurs with probability $ \left(N-1-j\right)/\left(N-1\right)$. 
	The right-hand side on the sixth line of Eq.~\eqref{eq:star-complete-T2} is equal to the product of these two probabilities.
	One can similarly derive the other lines of Eq.~\eqref{eq:star-complete-T2}. It should be noted that
	single-step moves from $\left(1,j\right)$ to $\left(1,j-1\right)$ and those from 
	$\left(0,j\right)$ to $\left(0,j+1\right)$ are possible in $G_2$, whereas they do not occur in $G_1$.
	
	In Fig.~\ref{fig:unified_starcomplete}(a), we plot the fixation probability as a function of $r$ for switching network $(G_1, G_2, \tau)$ in which $G_1$ is the star graph and $G_2$ is the complete graph on four nodes. In this figure, we compare $(G_1, G_2, \tau)$ with $\tau =1$, $10$, and $50$, the static star graph, the aggregate network, and the Moran process. Figure~\ref{fig:unified_starcomplete}(a) indicates that $(G_1, G_2, 10)$ and $(G_1, G_2, 50)$ are amplifiers \add{of selection} and that $(G_1, G_2, 1)$ is a suppressor. We plot the difference in the fixation probability between the switching networks and the Moran process in Fig.~\ref{fig:unified_starcomplete}(b). When $\tau=1$, the difference is positive for $r<1$ and negative for $r>1$, which verifies that $(G_1, G_2, 1)$ is a suppressor. This result is surprising because $G_1$ is an amplifier \add{of selection} and $G_2$ is equivalent to the Moran process and therefore not a suppressor \add{of selection}. In contrast, when $\tau=10$ and $\tau=50$, the difference from the Moran process is negative for $r<1$ and positive for $r>1$, which verifies that $(G_1, G_2, 10)$ and $(G_1, G_2, 50)$ are amplifiers \add{of selection}. The result for $\tau=50$ is close to that for the star graph. This is presumably because the first $\tau=50$ steps with $G_1$ are sufficient to induce fixation with a high probability given the small network size (i.e., $N=4$). 
	
	Figures~\ref{fig:unified_starcomplete}(a) and \ref{fig:unified_starcomplete}(b) also indicate that the aggregate network is a weak suppressor \add{of selection}. However, the aggregate network is a considerably weaker suppressor \add{of selection} than $(G_1, G_2, 1)$. Therefore, we conclude that the suppressing effect of the switching network mainly originates from the time-varying nature of the network rather than the structure of the weighted aggregate network.
	
	We show in Figs.~\ref{fig:unified_starcomplete}(c) and ~\ref{fig:unified_starcomplete}(d) the fixation probability and its difference from the case of the Moran process, respectively, as a function of $r$ for $N=50$. We observe that the switching network is an amplifier \add{of selection} for all the values of $\tau$ that we considered, i.e., $\tau=1$, $10$, and $50$. In contrast, the aggregate network is a suppressor \add{of selection} albeit an extremely weak one. The amplifying effect of the switching network is stronger for a larger value of $\tau$. Unlike in the case of four nodes (see Figs.~\ref{fig:unified_starcomplete}(a) and \ref{fig:unified_starcomplete}(b)), the switching networks with 50 nodes are far less amplifying than the star graph even with $\tau=50$. This phenomenon is expected because fixation in a static network with 50 nodes usually needs much more than 50 steps.
	
	These results for the switching networks with $N=4$ and $N=50$ nodes remain similar for $(G_2, G_1, \tau)$, i.e., when we swap the order of 
	$G_1$ and $G_2$ (see Figs.~\ref{fig:reversed_order_star_com_bipartite}(a) and \ref{fig:reversed_order_star_com_bipartite}(b)).
	
	The present switching network is a suppressor \add{of selection} when $N=4$ and $\tau=1$ and an amplifier \add{of selection} when $N=50$ or $\tau \in \{10, 50\}$.
	To examine the generality of these results with respect to the number of nodes, $N$,
	we show in Figs.~\ref{fig:unified_starcomplete}(e) and ~\ref{fig:unified_starcomplete}(f) the fixation probability relative to that for the Moran process
	at $\tau=1$ and $\tau=50$, respectively, as a function of $N$. In both figures, we show the fixation probabilities at $r=0.9$ and $r=1.1$. Figure~\ref{fig:unified_starcomplete}(e) indicates that the switching network is a suppressor \add{of selection} for $N \leq 4$ and an amplifier \add{of selection} for $N \ge 5$ when $\tau = 1$.
	We have confirmed that this switching network with $N=3$ nodes is a suppressor \add{of selection} by calculating the fixation probability across a range of $r$ values in (see Fig.~\ref{fig:suppressors_additional}(a) \add{ in \hyperref[sec:a2-further-examples-of-small-amplifying-switching-networks-in-which-g1-is-the-star-graph]{Appendix B}).}
	Figure~\ref{fig:unified_starcomplete}(f) indicates that $(G_1, G_2, 50)$ is an amplifier \add{of selection} for any $N$.
	
	\begin{figure}
		\begin{center}
			\includegraphics[width=14cm]{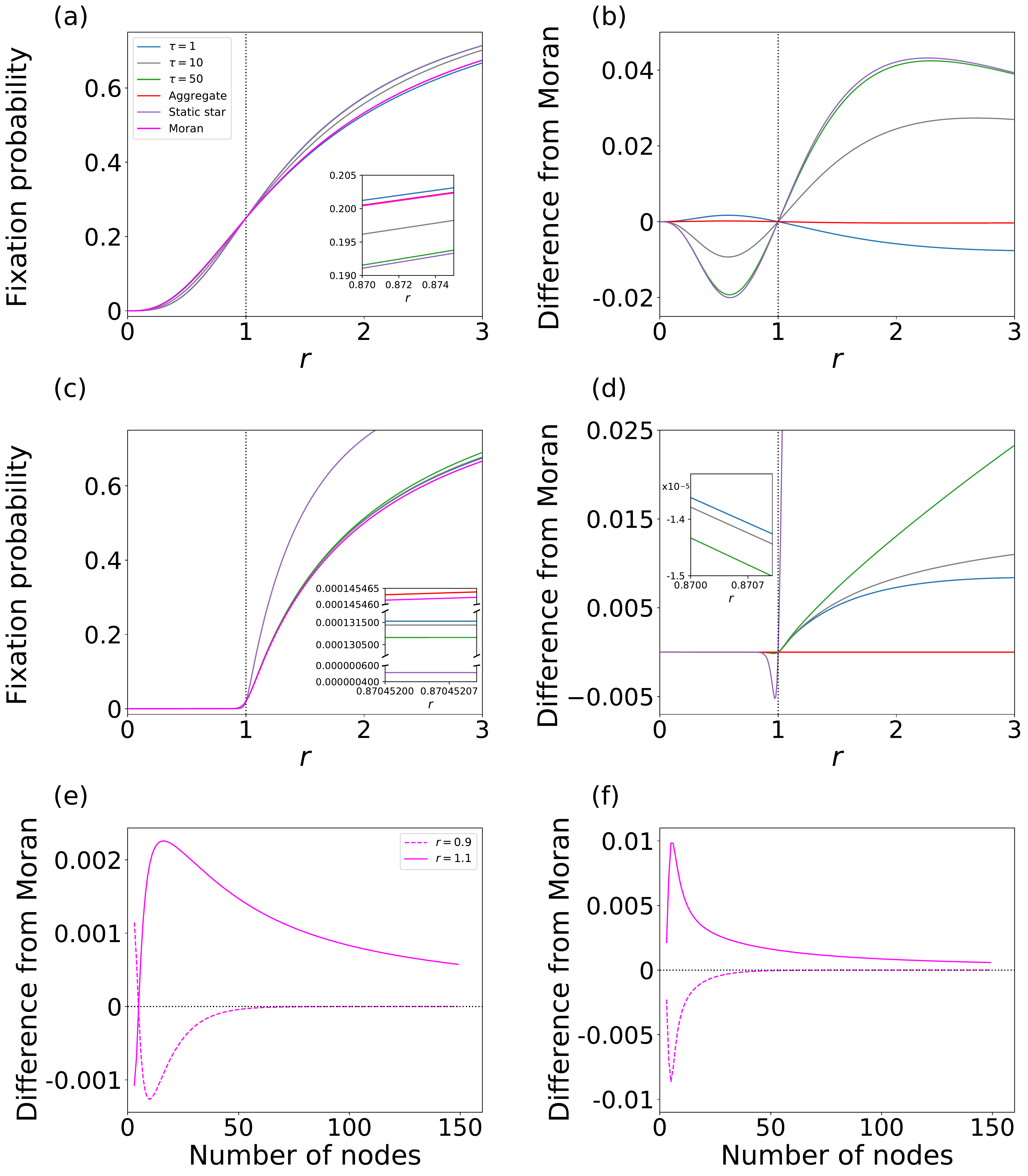}
			\caption{Fixation probability for switching networks in which $G_1$ is the star graph and $G_2$ is the complete graph. (a) Fixation probability for $N=4$. (b) Difference in the fixation probability from the case of the Moran process for $N=4$. (c) Fixation probability for $N=50$.  (d) Difference in fixation probability from the case of the Moran process for $N=50$. In (a)--(d), we also show the results for $G_1$ (i.e., star graph) and the aggregate network, and the vertical lines at $r=1$ are a guide to the eyes. The insets magnify selected ranges of $r<1$. (e) and (f): Difference in the fixation probability for the switching network relative to the Moran process as a function of $N$ at $r=0.9$ and $1.1$. We set $\tau=1$ in (e) and $\tau=50$ in (f). In (e) and (f), the smallest value of $N$ is three.}
			\label{fig:unified_starcomplete}
		\end{center}
	\end{figure}

	\subsubsection{Combination of the star graph and the complete bipartite graph}
	
	In this section, we analyze the switching network in which $G_1$ is the star graph and $G_2$ is the complete bipartite graph $K_{N_{1},N_{2}}$. By definition,
	$K_{N_{1}, N_{2}}$ has two disjoint subsets of nodes $V_1$ and $V_2$, and $V_1$ and $V_2$ contain $N_1$ and $N_2 $ nodes, respectively. Every node in $V_1$ is adjacent to every node in $V_2$ by an edge. Therefore, every node in $V_2$ is adjacent to every node in $V_1$. Without loss of generality, we assume that the hub node in $G_1$ is one of the $N_1$ nodes in $V_1$. 
	
	Because of the symmetry, we do not need to distinguish among the $N_1 - 1$ nodes that are leaf nodes in $G_1$ and belong to $V_1$ in $G_2$, or among the $N_2$ nodes that belong to $V_2$ in $G_2$. Therefore, one can specify the state of this switching network by a tuple $(i, j, k)$, where $i \in \{0,1 \}$ represents whether the hub is occupied by a resident, corresponding to $i=0$, or mutant, corresponding to $i=1$; variable $j \in \{0, \ldots, N_1 - 1\}$ represents the number of mutants among the $N_1 - 1$ nodes that are leaves in $G_1$ and belong to $V_1$ in $G_2$; variable $k \in \{0, \ldots, N_2 \}$ represents the number of mutants among the $N_2$ nodes in $V_2$. Tuples $(0, 0, 0)$ and $(1, N_1-1, N_2)$ correspond to the fixation of the resident and mutant type, respectively. Using this representation of the states, we reduce the $2^N \times 2^N$ transition probability matrix to a $2N_{1}\left(N_{2}+1\right)\times 2N_{1}\left(N_{2}+1\right)$ transition probability matrix. \add{We show the transition probabalities $T^{(1)}$ and $T^{(2)}$ in
	\hyperref[sec:adda3-transition-probabilities-for-the-combination-of-the-star-graph-and-the-complete-bipartite-graph]{Appendix C}.}

	\del{The transition probability matrix for the star graph is given by}
	\begin{equation}\del{
		T^{\left(1\right)}_{\left(i,j,k\right)\rightarrow\left(i',j',k'\right)} = \begin{cases}
			\frac{r(j+k)}{C_3} & \text{if } i=0 \text{ and } i'=1,
			\\[1mm]
			\frac{N-1-j-k}{C_4} & \text{if } i = 1 \text{ and } i'=0,
			\\[1mm]
			\frac{1}{C_3}\cdot\frac{j}{N-1} & \text{if } i=0 \text{ and }j'= j-1,
			\\[1mm]
			\frac{1}{C_3}\cdot\frac{k}{N-1} & \text{if }i=0 \text{ and }k' = k-1,
			\\[1mm]
			\frac{r}{C_4}\cdot\frac{N_{1}-1-j}{N-1} & \text{if }i=1 \text{ and } j' = j+1,
			\\[1mm]
			\frac{r}{C_4}\cdot\frac{N_2-j}{N-1} & \text{if }i=1 \text{ and } k' = k+1,
			\\[1mm]
			1-\sum\limits_{\mathclap{\substack{(i'',j'',k'')\neq \\ (i,j,k)}}}T^{\left(1\right)}_{(i,j,k) \rightarrow (i'',j'',k'')} & \text{if } (i',j',k')=(i, j, k),
			\\[1mm]
			0 & \text{ otherwise,}
		\end{cases}}
		\label{eq:star-bipartite-T1}
	\end{equation}
	\del{where $C_3 = r(j+k)+(N-j-k)$ and $C_4 = r\left(j+k+1\right)+\left(N-j-k-1\right)$.} 
	\del{The first line of Eq.~\eqref{eq:star-bipartite-T1} represents the probability that the type of the hub changes from the resident to mutant. For this event to occur, one of the $j+k$ leaf nodes occupied by the mutant must be chosen as parent, which occurs with probability $r (j+k) / C_3$. Then, because any leaf node is only adjacent to the hub node, the hub node is always selected for death. Therefore, the probability of $i$ changing from $0$ to $1$ is equal to $r(j+k) / C_3$. As another example, consider state $\left(1,j,k\right)$. For the state to change from $\left(1,j,k\right)$ to $ \left(1,j+1,k\right)$, the hub node, which the mutant type currently inhabits, must be selected as parent with probability $r/C_4$. Then, one of the $j$ leaf nodes of the resident type in $V_1$ must be selected for death, which occurs with probability $\left[(N_1- 1) - j\right]/\left(N-1\right)$. The fifth line of Eq.~\eqref{eq:star-bipartite-T1} is equal to the product of these two probabilities. One can similarly derive the other lines of Eq.~\eqref{eq:star-bipartite-T1}.}
	
	\del{The transition probability matrix for the complete bipartite graph is given by}
	\begin{equation}
		\del{T^{\left(2\right)}_{\left(i,j,k\right)\rightarrow\left(i',j',k'\right)} = \begin{cases}
			\frac{r k}{C_3}\cdot\frac{1}{N_1} & \text{if }i=0 \text{ and } i'=1,
			\\[1mm]
			\frac{N_2-k}{C_4}\cdot\frac{1}{N_1} & \text{if }i=1 \text{ and } i'=0,
			\\[1mm]
			\frac{N_2-k}{C_3}\cdot\frac{j}{N_{1}} & \text{if } i=0 \text{ and } j' = j-1,
			\\[1mm]
			\frac{r k}{C_3}\cdot\frac{N_{1}-1-j}{N_{1}} & \text{if } i=0 \text{ and } j' = j +1,
			\\[1mm]
			\frac{N_{1}-j}{C_3}\cdot\frac{k}{N_{2}} & \text{if } i=0 \text{ and } k' = k-1,
			\\[1mm]
			\frac{rj}{C_3}\cdot\frac{N_{2}-k}{N_{2}}  & \text{if } i=0 \text{ and } k' = k+1,
			\\[1mm]
			\frac{N_{2}-k}{C_4}\cdot\frac{j}{N_{1}} & \text{if }i=1 \text{ and } j' = j-1,
			\\[1mm]
			\frac{rk}{C_4}\cdot\frac{N_{1}-1-j}{N_{1}} & \text{if }i=1 \text{ and } j' = j+1,
			\\[1mm]
			\frac{N_1-1-j}{C_4}\cdot \frac{k}{N_2} & \text{if }i=1 \text{ and } k' = k-1,
			\\[1mm]
			\frac{r\left(j+1\right)}{C_4}\cdot\frac{N_{2}-k}{N_2}  & \text{if }i=1 \text{ and } k'=k+1,
			\\[1mm]
			1-\sum\limits_{\mathclap{\substack{(i'',j'',k'')\neq \\ (i,j,k)}}}T^{\left(2\right)}_{\left(i,j,k\right)\rightarrow\left(i'',j'',k''\right)}  & \text{if } (i',j',k')=(i,j,k).
			\\[1mm]
			0 & \text{ otherwise.}
		\end{cases}}
		\label{eq:star-bipartite-T2}		
	\end{equation}
	\del{The first line of Eq.~\eqref{eq:star-bipartite-T2} represents the probability that the type of the hub changes from the resident to mutant. For this event to occur, one of the $k$ mutant nodes in $V_2$ must be selected as parent with probability $rk/C_3$. Then, the hub node must be selected for death with probability $1/N_1$ because each node in $V_2$ is only adjacent to all the $N_1$ nodes in $V_1$. Therefore, the probability of $i$ changing from $0$ to $1$ is equal to $\left(rk/
	C_3\right)\cdot \left(1/N_1\right)$.
	As another example, consider state $\left(1,j,k\right)$, in which there are $j+k+1$ mutants in total. For the state to change from $\left(1,j,k\right)$ to $ \left(1,j+1,k\right) $, one of the $k$ mutant nodes in $V_2$ must first be selected as parent with probability $rk/C_4$. Then, one of the $j$ leaf nodes in $V_1$ of the resident type must be selected for death, which occurs with probability $(N_1-1-j)/N_1$. The eighth line of Eq.~\eqref{eq:star-bipartite-T2} is equal to the product of these two probabilities. One can similarly derive the other lines of Eq.~\eqref{eq:star-bipartite-T2}.}

	In Figs.~\ref{fig:unified_bipartite_plots}(a) and \ref{fig:unified_bipartite_plots}(b), we show the fixation probability and its difference from the case of the Moran process, respectively, for the switching network in which $G_1$ is the star on $N=4$ nodes and $G_2$ is the complete bipartite graph $K_{N_1, N_2} $ with $N_1 = N_2 = 2$. We set $\tau=1,10$, and $50$, and varied $r$. We also show the results for
	$G_1$, $G_2$, and the aggregate network in these figures for comparison.
	We find that $(G_1, G_2, 1)$ is a suppressor \add{of selection}. In contrast, $G_1$ is an amplifier \add{of selection}, and $G_2$ is neutral (i.e., equivalent to the Moran process).
	In fact, no static unweighted network with five nodes or less is a suppressor \add{of selection} \cite{alcalde2017suppressors}. Because the aggregate network is an amplifier \add{of selection}, albeit a weak one, the suppressing effect of $(G_1, G_2, 1)$ owes to the time-varying nature of the switching network.
	Similar to the case in which $G_2$ is the complete graph shown in Fig.~\ref{fig:unified_starcomplete}, $(G_1, G_2, 10)$ and $(G_1, G_2, 50)$ are amplifiers \add{of selection}, and the behavior of $(G_1, G_2, 50)$ is close to that for $G_1$, i.e., the star graph.
	
	In Figs.~\ref{fig:unified_bipartite_plots}(c) and \ref{fig:unified_bipartite_plots}(d), we show the fixation probability and its difference from the case of the Moran process, respectively, for $N_1 = N_2 = 20$. \add{We have set $N = N_1 + N_2 = 40$ as opposed to $N = 50$, which we used for the switching network analyzed in section~\ref{sub:star-complete}, because of the computational cost.} In contrast to the case of $N_1 = N_2 = 2$, the switching network with $N_1 = N_2 = 20$ is an amplifier \add{of selection} for the three values of $\tau$. Furthermore, in contrast to when $N_1 = N_2 = 2$, the fixation probabilities for the switching networks are closer to those for the Moran process than to those for the star graph. \add{To explore the case $N_1 \neq N_2$, we show the results for $\left(N_1,N_2\right)=\left(4,2\right)$ and $\left(N_1,N_2\right)=\left(30,10\right)$ in \hyperref[sec:addd-examples-in-which-g1-is-the-star-graph-g2-is-the-complete-bipartite-graph-and-n1neq-n2]{Appendix D}. The switching network with $\left(N_1,N_2\right)=\left(4,2\right)$ is neither amplifier or suppressor of selection. However, its fixation probabilities are close to those for the Moran process than to those for the star or bipartite complete graph. The switching network with $\left(N_1,N_2\right)=\left(30,10\right)$ is an amplifier of selection and behaves similarly to the switching network with $\left(N_1, N_2\right) = \left( 20, 20 \right)$.}
%
% \nm{Because I moved this new text here, please reorder the figure/appendix so that they appear in the order appearing in the main text. Also, I added the description of the new figure. Please check the accuracy.}\jb{It's correct. The appendix ordering doesn't change since A1, A2, A3 have been referenced before here.}
	
	These results for the switching networks with $N=4$ and $N=40$ nodes remain similar for switching networks $(G_2, G_1, \tau)$, as we show in Figs.~\ref{fig:reversed_order_star_com_bipartite}(c) and \ref{fig:reversed_order_star_com_bipartite}(d).
	
	To examine the dependence of the fixation probability on the number of nodes, we show in Fig.~\ref{fig:unified_bipartite_plots}(e) the difference between the fixation probability for the present switching network and that for the Moran process as we vary $N$. We set $\tau=1$ and $N_1=N_2 = N/2 \ge 2$, and compute the fixation probability at $r=0.9$ and $r=1.1$. Figure~\ref{fig:unified_bipartite_plots}(e) indicates that the switching network is a suppressor \add{of selection} only when $N_1=N_2=2$ (i.e., $N=4$) and amplifier \add{of selection} for any larger $N$. When we allow $N_1 \neq N_2$, we found just one additional suppressor \add{of selection} apart from $(N_1, N_2) = (2, 2)$ under the constraints $\tau=1$ and $2\le N_1, N_2 \le 10$, which is $(N_1, N_2) = (3, 2)$ (see Fig.~\ref{fig:suppressors_additional}(b) \add{in \hyperref[sec:adda3-transition-probabilities-for-the-combination-of-the-star-graph-and-the-complete-bipartite-graph]{Appendix B}}). With $\tau = 50$, this switching network is amplifier \add{of selection} for any $N$ (see Fig.~\ref{fig:unified_bipartite_plots}(f)).
	
	\begin{figure}
		\begin{center}
			\includegraphics[width=14cm]{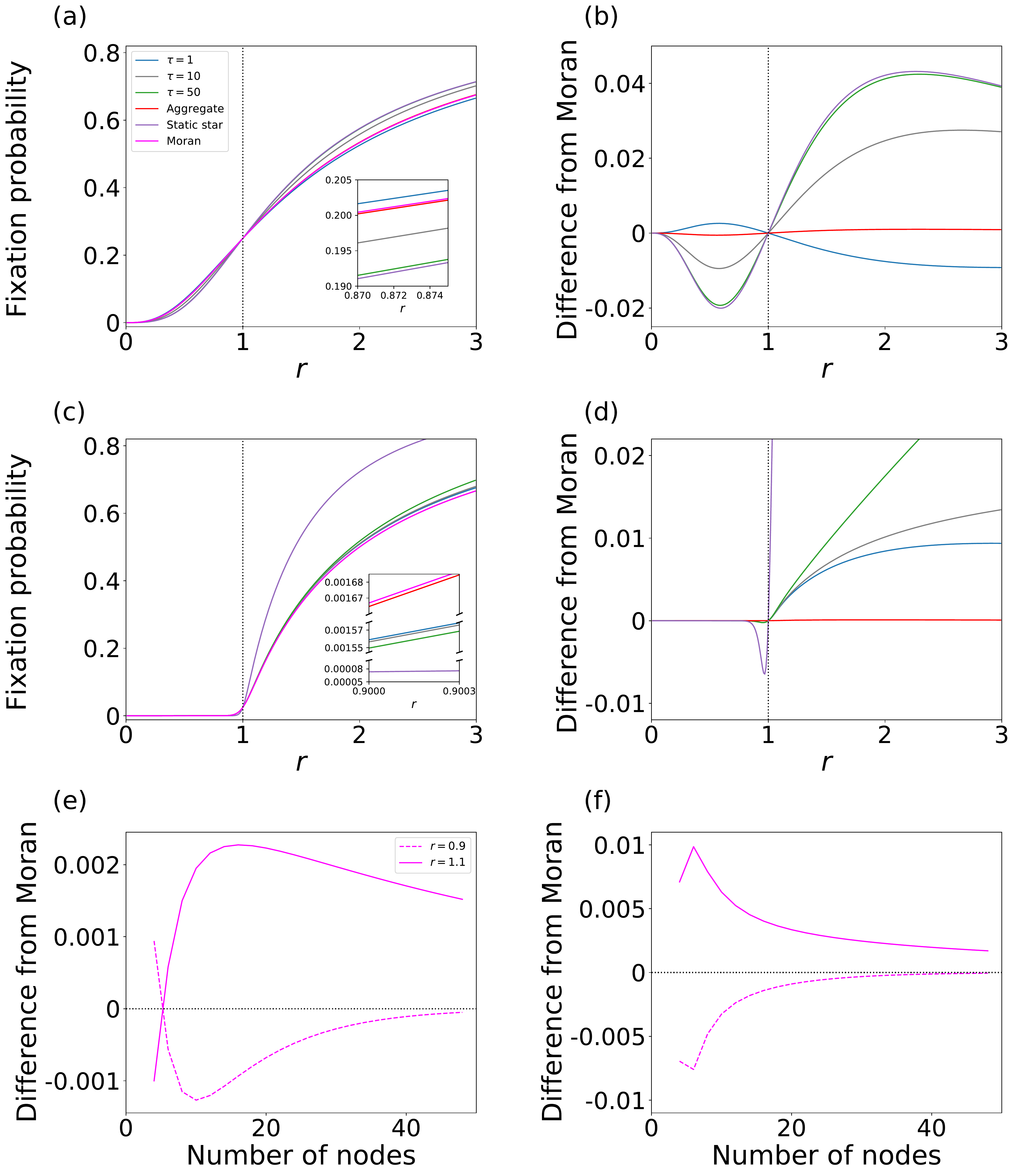}
			\caption{Fixation probability for switching networks in which $G_1$ is the star graph and $G_2$ is the complete bipartite graph. (a) Fixation probability for $N_1=N_2=2$. (b) Difference in the fixation probability from the case of to the Moran process for $N_1=N_2=2$. (c) Fixation probability for $N_1=N_2=20$. (d) Difference in the fixation probability from the case of the Moran process for $N_1=N_2=20$. (e) and (f): Difference in the fixation probability for the switching network relative to the Moran process as a function of $N$ at $r=0.9$ and $1.1$. We set $\tau=1$ in (e) and $\tau=50$ in (f). In (e) and (f), the smallest value of $N$ is four.}
			\label{fig:unified_bipartite_plots}
		\end{center}
	\end{figure}
	
	\subsection{Empirical temporal networks}
	
	\subsubsection{Construction of switching networks}\label{sec:construction-of-switching-networks}
	
	Finally, we numerically simulate the birth-death process on four switching networks informed by empirical temporal network data. We split each of the temporal network data set into two static networks $(V_1, E_1)$ and $(V_2, E_2)$, where $(V_1, E_1)$ contains the first half of the time-stamped edges in terms of the time, $(V_2, E_2)$ containing the second half of the time-stamped edges, $V_1$ and $V_2$ are sets of nodes, and $E_1$ and $E_2$ are sets of edges.
	For simplicity, we regard $(V_1, E_1)$ and $(V_2, E_2)$ as unweighted networks.
	 \add{We note that the purpose of studying these empirical networks is not to examine how fixation occurs in real contact sequences but to explore the generality of the results obtained in the previous sections in asymmetric and large switching networks.}

	For two of the four empirical switching networks, both $V_1$ and $V_2$ contain all nodes. In this case, we switch between $G_1 \equiv (V_1, E_1)$ and $G_2 \equiv (V_2, E_2)$. For the other two empirical switching networks, either $V_1$ or $V_2$ misses some nodes in the original temporal network. In this case, we construct switching networks in the following two manners. With the first method, we only use the nodes in $V_1 \cap V_2$ and the edges that exist between pairs of nodes in $V_1 \cap V_2$ as $G_1$ and $G_2$. For each of the two empirical data sets for which $V_1$ or $V_2$ misses some nodes, we have confirmed that
	the first and second halves of the static networks induced on $V_1\cap V_2$ created with this method are connected networks.
	With the second method, we use all nodes for both $G_1$ and $G_2$. In other words, we set $G_1 = (V_1 \cup V_2, E_1)$ and $G_2 = (V_1 \cup V_2, E_2)$. Therefore, if $v \in V_1$ and $v \notin V_2$, for example, then $v$ is an isolated node in $G_2$. Except with special initial conditions, the fixation of either type never occurs in a static network with isolated nodes. However, the fixation does occur in the switching network if the aggregate network is connected, which we have confirmed to be the case for all our empirical data sets.
		
	\subsubsection{Simulation procedure}
	
	As the initial condition, we place a mutant on one node selected uniformly at random and all the other $N-1$ nodes 
	are of the resident type. Then, we run the birth-death process until all nodes were of the same type. We carried out $2 \times 10^5$ such runs in parallel on $56$ cores, giving us a total of $112 \times 10^5$ runs, for each network and each value of $r$. We numerically calculated the fixation probability as the fraction of runs in which the mutant fixates. We simulated the switching networks with $\tau \in \{1, 10, 50\}$ and $r \in \{0.7,0.8,0.9,1,1.1,1.2,1.3,1.4,1.5,1.6,1.7\}$ for all the networks except the hospital network of $75$ nodes. For the hospital network, we omitted $r=1.6$ and $1.7$ due to high computational cost.
	
	\subsubsection{Data}
	
	The ants' colony data, which we abbreviate as ant \cite{ants2015quevillonsocial}, has 39 nodes and 330 time-stamped edges. Each node represents an ant in a colony. An edge represents a trophallaxis event, which was recorded when the two ants were engaged in mandible-to-mandible contact for greater than one second. The first and second halves of the data have 34 nodes each.
	
	The second data is the contacts between members of five households in the Matsangoni sub-location within the Kilifi Health and Demographic Surveillance Site (KHDSS) in coastal Kenya~\cite{kilifi2016quantifying}. A household was defined as the group of individuals who ate from the same kitchen~\cite{kilifi2016quantifying}. Each participant in the study had a wearable sensor that detected the presence of another sensor within approximately 1.5 meters. Each node is an individual in a household. An edge represents a time-stamped contact between two individuals. There were 47 nodes. There were 219 time-stamped edges representing contacts between pairs of individuals in different households and $32,426$ time-stamped edges between individuals of the same households. Both the first and second halves contain all the 47 nodes and are connected networks as static network owing to the relatively large number of time-stamped edges.
	
	The third data is a mammalian temporal network based on interaction between raccoons~\cite{reynolds2015raccoon}. A node represents a wild raccoon. The time-stamped events were recorded whenever two raccoons came within approximately 1 to 1.5 meters for more than one second, using proximity logging collars that were placed on raccoons. The recording was made in Ned Brown Forest Preserve in suburban Cook County, Illinois, USA, from July 2004 to July 2005. There are 24 nodes and $2,000$ time-stamped edges. Both the first and second halves of the data contain all the 24 nodes and are connected networks as static network.
	
	The fourth data is a contact network in a hospital~\cite{hospitalvanhems2013estimating}. The data were recorded in a geriatric unit of a university hospital in Lyon, France, from December 6, 2010 at 1 pm to December 10, 2010 at 2 pm. The unit contained 19 out of the $1,000$ beds in the hospital. During the recording period, 50 professionals worked in the unit, and 31 patients were admitted. Fourty-six among the 50 professionals and 29 among the 31 patients participated in the study. Therefore, the network had 75 nodes in total. The professionals comprised of 27 nurses or nurses’ aides, 11 medical doctors, and 8
	administrative staff members. An edge represents a time-stamped contact between two individuals; there are $32,424$ time-stamped edges. The first and second halves of the data contain 50 nodes each.
	
	We obtained the ant, raccoon, and hospital data from \url{https://networkrepository.com/}~\cite{nr-aaai15}. We obtained the Kilifi data from \url{http://www.sociopatterns.org/}.
	
	\subsubsection{Numerical results}
	
	We investigate the fixation probability on the switching networks with $\tau=1$, $10$, and $50$, static networks $G_1$ and $G_2$, and the aggregate network. We remind that the aggregate network is a static weighted network, whereas $G_1$ and $G_2$ are unweighted networks. For the ant and hospital data,
	the switching networks constructed with the second method are different from those constructed with the first method. For these two data sets, fixation does not occur on $G_1$ and $G_2$ because they miss some nodes. Therefore, we do not analyze the fixation probability on $G_1$ and $G_2$ for these data sets.
	
	We show in Figs.~\ref{fig:empirical_figures}(a) and \ref{fig:empirical_figures}(b) the fixation probability on the ant switching networks constructed with the first and second methods, respectively. Because we are interested in whether the switching networks are amplifiers or suppressors \add{of selection}, we only show the difference between the fixation probability on the given network and that for the Moran process in Fig.~\ref{fig:empirical_figures}. Figure~\ref{fig:empirical_figures}(a) indicates that the switching networks are amplifiers \add{of selection} but less amplifying than each of its constituent static networks, $G_1$ and $G_2$. Another observation is that the fixation probability on the static aggregate network is close to that on the switching networks. In this sense,
	the switching networks do not yield surprising results. The switching networks are more strongly amplifying when $\tau$ is larger. Moreover, the fixation probability on the switching network is closer to that on $G_1$ when $\tau$ is larger. This result is expected because the evolutionary dynamics is the same between the switching networks and $G_1$ in the first $\tau$ time steps.
	For the switching networks constructed with the second method, Fig.~\ref{fig:empirical_figures}(b) shows that the switching networks are amplifiers and more amplifying than the static aggregate network. This result is qualitatively different from that for the switching networks constructed with the first method shown in Fig.~\ref{fig:empirical_figures}(a).
	
	We show the results for the Kilifi networks in Fig.~\ref{fig:empirical_figures}(c). Because the first and second methods yield the same $G_1$ and $G_2$ for the Kilifi data, we only present the results for the first method for this data set and also for the next one (i.e., racoon networks). The figure indicates that the switching networks are amplifiers but less amplifying than $G_1$ and $G_2$ and similarly amplifying compared to the aggregate network. These results are similar to those for the ant networks shown in Fig.~\ref{fig:empirical_figures}(a).
	
	We show the results for the raccoon networks in Fig.~\ref{fig:empirical_figures}(d). We find that the switching networks are amplifiers but less amplifying than $G_1$ and $G_2$, similar to the case of the ant and Kilifi networks. We also find that the switching networks are more amplifying than the aggregate network.
	
	We show the results for the hospital switching networks in Figs.~\ref{fig:empirical_figures}(e) and \ref{fig:empirical_figures}(f). The results for the switching networks constructed with the first method (see  Fig.~\ref{fig:empirical_figures}(e)) are
	similar to those for the raccoon networks shown in Fig.~\ref{fig:empirical_figures}(d).
	The switching networks constructed with the second method (see Fig.~\ref{fig:empirical_figures}(f)) are more amplifying than the aggregate network, similar to the case of the ant networks generated by the same method (see Fig.~\ref{fig:empirical_figures}(b)).
	
	In sum, for these empirical temporal networks, we did not find a surprising result that the fixation probability for the switching networks is not an interpolation of those for the two static networks $G_1$ and $G_2$. However, the fixation probability for the empirical switching networks depends on the $\tau$ value and deviates from the prediction from the aggregate network in multiple ways.

	%%%%%%%%%%%%%%%%%%%%%%%
	%%%%%%%%%%%%%%%%%%%%%%%
	
	\begin{figure}
		\begin{center}
			\includegraphics[width=13cm]{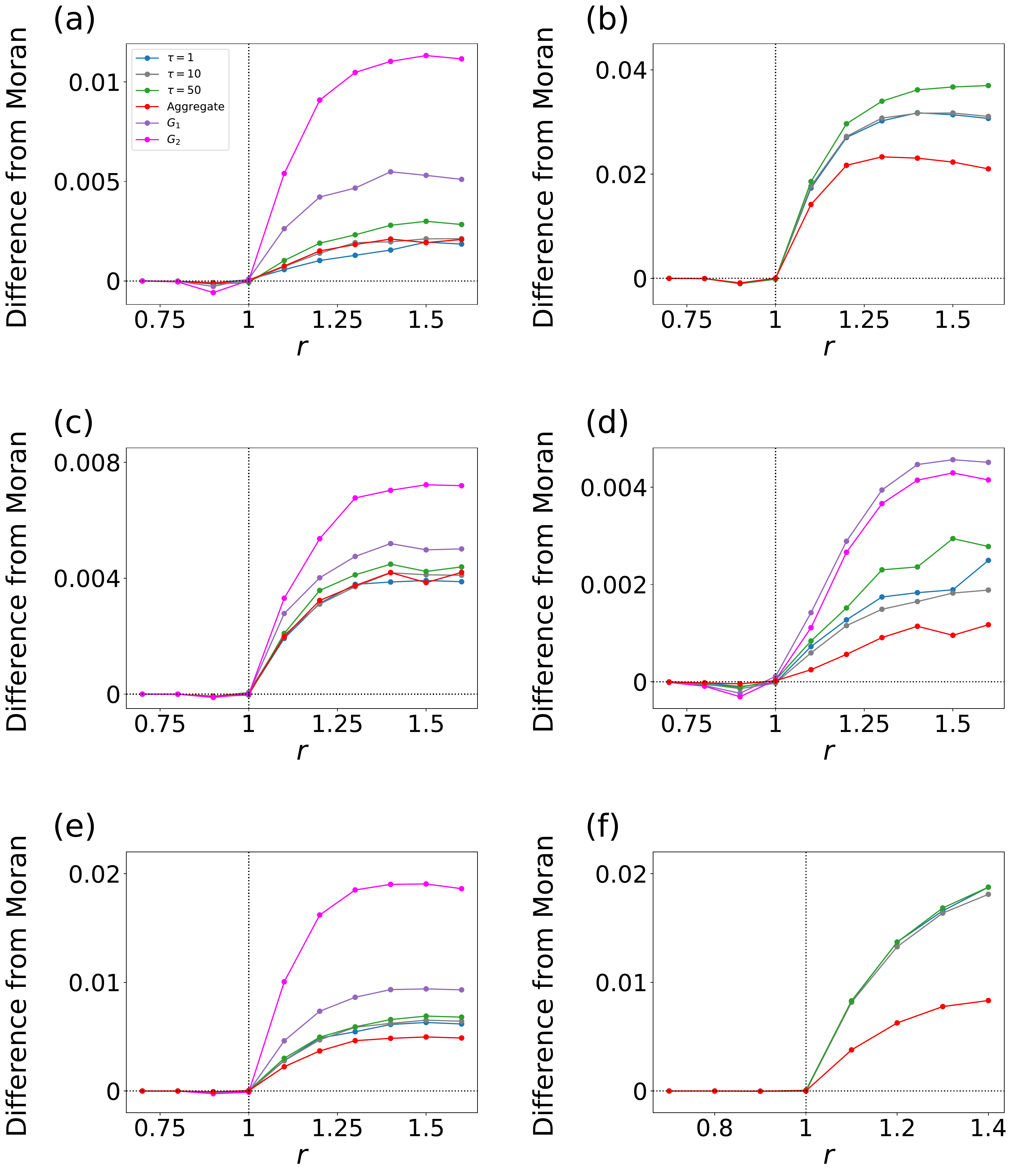}
			
			\caption{Fixation probability on empirical switching networks. In each panel, we show the difference in the fixation probability from the case of the Moran process as a function of $r$. (a) Ant networks constructed with the first method. (b) Ant networks constructed with the second method. (c) Kilifi switching networks. (d) Raccoon networks. (e) Hospital networks constructed with the first method. (f) Hospital networks constructed with the second method. We compared the fixation probability on switching networks with $\tau \in \{1, 10, 50\}$, $G_1$, $G_2$, and the aggregate network in each panel.}
			\label{fig:empirical_figures}
		\end{center}
		
	\end{figure}
	%%%%%%%%%%%%%%%%%%%%%
	%%%%%%%%%%%%%%%%%%%%%
	
	\section{Discussion}
	
	We have shown that, under the birth-death updating rule and uniform initialization, a majority of the switching networks on six nodes are suppressors of \del{natural} selection.
	This result contrasts with the case of static networks, for which there exists only one suppressor \add{of selection} on six nodes~\cite{alcalde2017suppressors}.
	We also found that switching networks alternating between the star graph and the complete graph and those alternating between the star graph and the complete bipartite graph are suppressors \add{of selection} when the number of nodes, $N$, is small. When $N$ is larger, the same switching networks are amplifiers \add{of selection} but less amplifying than the star graph. Among the empirical networks that we analyzed, we did not find any suppressors. However, these switching networks were notably less amplifying than the constituent static networks $G_1$ and $G_2$. In fact, the less amplifying nature of switching networks is largely explained by the aggregate weighted network, or the static network obtained by the superposition of $G_1$ and $G_2$. Therefore, our results for the empirical switching networks are not surprising.
	The result that the switching network composed of two amplifying static networks can be suppressor is our main finding. Because all the instances that we have found are small networks, searching suppressing switching networks with larger $N$ including systematically constructing such instances remains future work.
	
	\add{Our choices of the larger networks are primarily driven by computational feasibility. The complete graph, star graph, and complete bipartite graph are convenient families of networks owing to their highly symmetric nature, which drastically reduces the number of the unknowns to be determined. Similarly, all the empirical networks that we used had at most 75 nodes due to computational cost. Additionally, we avoided disconnected and sparse networks because fixation requires a network to be connected, and splitting a sparse network into two networks often resulted in disconnected components. Nevertheless, by studying small networks, larger symmetric networks, and the empirical examples, we tried to provide a broader picture of the evolutionary dynamics on switching networks. However, there remains ample room for future work. For instance, faster algorithms for approximate computation for larger switching networks, such as those assuming weak selection~\cite{allen2021fixation}, remain to be explored. We could also attempt to reduce simulation times. In \cite{chatterjee2017faster}, instead of sampling every time step of the evolutionary dynamics, only `effective' steps are sampled. Effective steps are defined as those in which the network state changes. The steps in which a resident replaces a resident or a mutant replaces a mutant are deemed as ineffective steps, which one does not sample in their algorithm, hence accelerating the overall simulation time.} Additionally, we studied switching networks with only two snapshots. It is straightforward to extend the present computational framework to the case of switching networks with more than two snapshots. Last, many temporal network data are provided as a list of time-stamped events between pairs of nodes. Evolutionary dynamics driven by such event-based temporal network data is also worth studying.
	
	We considered exogenous changes of the network over time in this study. Another opportunity of research is to assume that the change of the network structure over time is driven by the state of the system, which is referred to as adaptive networks \cite{gross2008adaptive, sayama2013modeling}. The recent modeling framework inspired by biological examples in which the residents and mutants use different static networks defined on the same node set \cite{tkadlec2021natural,melissourgos2022extension} can be interpreted as an example of fixation dynamics on adaptive networks. Allowing nodes to stochastically sever and create edges they own as the node's type flips from the resident to mutant and vice versa may lead to new phenomena in fixation dynamics. Such models have been extensively studied for evolutionary games on dynamic networks~\cite{santos2006cooperation,pacheco2006coevolution,fu2009partner}. 
	
	We recently found that most hypergraphs are suppressors \add{of selection} under the combination of a birth-death process and uniform initialization, which are the conditions under which most of conventional networks are amplifiers \add{of selection}~\cite{Ruodan23}. It has been longer known that most undirected networks are suppressors \add{of selection} under the death-birth process~\cite{hindersin2015most}\add{,} and in directed networks under various imitation rules including birth-death processes \cite{masuda2009directionality}. 
	The degree of amplification and suppression also depends on the initialization~\cite{adlam2015amplifiers,pavlogiannis2018construction}. For example,
	non-uniform initializations can make the star, which is a strong amplifier \add{of selection} under the birth-death process and uniform initialization, a suppressor \add{of selection} \cite{adlam2015amplifiers}.
	Furthermore, it has been shown that the amplifiers \add{of selection} are transient and bounded~\cite{tkadlec2020limits}. 
	Our results suggest that small temporal networks are another major case in which suppressors \add{of selection} are common. These results altogether encourage us to explore different variants of network models and evolutionary processes to clarify how common amplifiers \add{of selection} are. This task warrants future research.
	
	% Death-birth rule~\cite{nowak_book,frean2008death,masuda2009directionality,hindersin2015most}. 
	
	% \section*{Acknowledgements}
	
	% We thank the SocioPatterns collaboration (\url{http://www.sociopatterns.org}) for providing the data set of contacts between household members in Kilifi, Kenya. 
	
	\section*{Funding}\label{sec:funding}
	N.M. acknowledges support from AFOSR European Office (under Grant No. FA9550-19-1-7024), the Japan Science and Technology Agency (JST) Moonshot R$\&$D (under Grant No. JPMJMS2021), and the National Science Foundation (under Grant No. 2052720 and 2204936).

	\section*{Appendices}\label{appendix}
	
	\subsection*{A. Switching networks \del{in which the first network is the star graph} \add{with the order of the static networks reversed and with random initialization time.}}\label{appendix:A1}
	
\add{In this section, we consider switching networks in which $G_2$ rather than $G_1$ appear first and those with uniformly random initialization time.}
	
\add{In Fig.~\ref{fig:reversed_order_star_com_bipartite}(a), we show the results for the six-node switching networks in which $G_1$ and $G_2$ are given in Fig.~\ref{fig:ampsbecomesup}(a). We find that $\left(G_2,G_1,1\right)$ and the switching network with $\tau=1$ and the random initialization time are both suppressors of selection. These variants of switching networks are amplifiers of selection when $\tau=10$ and $\tau=50$. These results are qualitatively the same as those for $\left(G_1,G_2, \tau \right)$.}
	
\add{Next,} \del{In this section,} we consider switching networks $(\del{G_1} \add{G_2}, \del{G_2} \add{G_1}, \tau)$ in which $G_1$ is the \add{star} \del{complete} graph and $G_2$ is the \del{star} \add{complete} graph. We show the difference in the fixation probability from the case of the Moran process for the switching networks with $N=4$ and $N=50$ in Figs.~\ref{fig:reversed_order_star_com_bipartite}(\del{a}\add{b}) and \ref{fig:reversed_order_star_com_bipartite}(\del{b}\add{c}), respectively. With $N=4$, we find that $(\del{G_1} \add{G_2}, \del{G_2} \add{G_1}, 10)$ and $(\del{G_1} \add{G_2}, \del{G_2} \add{G_1}, 50)$ are amplifiers \add{of selection}  and that
	$(\del{G_1} \add{G_2}, \del{G_2} \add{G_1}, 1)$ is a suppressor \add{of selection} (see Fig.~\ref{fig:reversed_order_star_com_bipartite}(\del{a}\add{b}).
	The aggregate network is a weak suppressor \add{of selection}. With $N=50$, we find that $(\del{G_1} \add{G_2}, \del{G_2} \add{G_1}, \tau)$ for all the three $\tau$ values (i.e., $\tau \in \{1, 10, 50\}$) are amplifiers \add{of selection} and that the aggregate network is a weak suppressor \add{of selection} (see Fig.~\ref{fig:reversed_order_star_com_bipartite}(\del{b}\add{c})). These results are qualitatively the same as those for \add{$(G_1, G_2, \tau)$} \del{the switching networks in which the order of $G_1$ and $G_2$ is the opposite,} shown in Fig.~\ref{fig:unified_starcomplete}. A main difference is that, when $\tau=50$, the fixation probability is reasonably close to that for the Moran process in the case of the present switching network because \add{the initially used static network, i.e.,} \del{$G_1$} \add{$G_2$,} is a regular graph and therefore equivalent to the Moran process. In contrast, in Fig.~\ref{fig:unified_starcomplete}, the switching network is much more amplifying because \del{$G_1$} \add{the initially used static network} is the star graph, which is a strong amplifier \add{of selection}. \add{As expected, the results in the case of the random initialization time are between those for $(G_1, G_2, \tau)$ and those for $(G_2, G_1, \tau)$.}
	
	We show in Figs.~\ref{fig:reversed_order_star_com_bipartite}(\del{c}\add{d}) and \ref{fig:reversed_order_star_com_bipartite}(\del{d}\add{e}) the results for \add{$(G_2, G_1, \tau)$} \del{the switching networks} with $N = 4$ and $N = 40$, respectively, in which $G_1$ is the \del{complete bipartite} \add{star} graph and $G_2$ is the \del{star} \add{complete bipartite} graph. 	With $N=4$, we find that $(\del{G_1} \add{G_1}, \del{G_2} \add{G_1}, 1)$ is a suppressor \add{of selection}, $(\del{G_1} \add{G_2}, \del{G_2} \add{G_1}, 10)$ and $(\del{G_1} \add{G_2}, \del{G_2} \add{G_1}, 50)$ are amplifiers \add{of selection} \del{$\tau=50$}, and the aggregate network is a weak amplifier \add{of selection} (see Fig.~\ref{fig:reversed_order_star_com_bipartite}(\del{c}\add{d})).
	With $N=40$, we find that $(\del{G_1} \add{G_2}, \del{G_2} \add{G_1}, \tau)$ with $\tau \in \{1, 10, 50\}$ is an amplifier \add{of selection} and that the aggregate network is a weak amplifier \add{of selection} (see Fig.~\ref{fig:reversed_order_star_com_bipartite}(\del{d}\add{e})). These results are similar to those for \add{$(G_1, G_2, \tau)$} \del{the switching networks in which the order of $G_1$ and $G_2$ is the opposite,} shown in Figs.~\ref{fig:unified_bipartite_plots}(a) and \ref{fig:unified_bipartite_plots}(b). 
	Similar to Figs.~\ref{fig:reversed_order_star_com_bipartite}(\del{a}\add{b}) and \ref{fig:reversed_order_star_com_bipartite}(\del{b}\add{c}), with $\tau=50$, the present switching networks are close in behavior to the Moran process because \add{the initially used static network, i.e., $G_2$,} \del{$G_1$} is a regular network. This result contrasts to the corresponding result for \add{$(G_1, G_2, 50)$} \del{the order-swapped switching network with $\tau=50$}, which is a relatively strong amplifier \add{of selection} because \add{the initially used static network} \del{$G_1$} is the star graph (see Figs.~\ref{fig:unified_bipartite_plots}(a) and \ref{fig:unified_bipartite_plots}(b)).
\add{Again, the results in the case of the random initialization time are between those for $(G_1, G_2, \tau)$ and those for $(G_2, G_1, \tau)$.}
	
	%%%%%%%%%%%%%%%%%%%%%%%
	
	%\setcounter{figure}{0}
	\add{\begin{figure}[!htbp]
		\begin{center}
			\includegraphics[width=13cm]{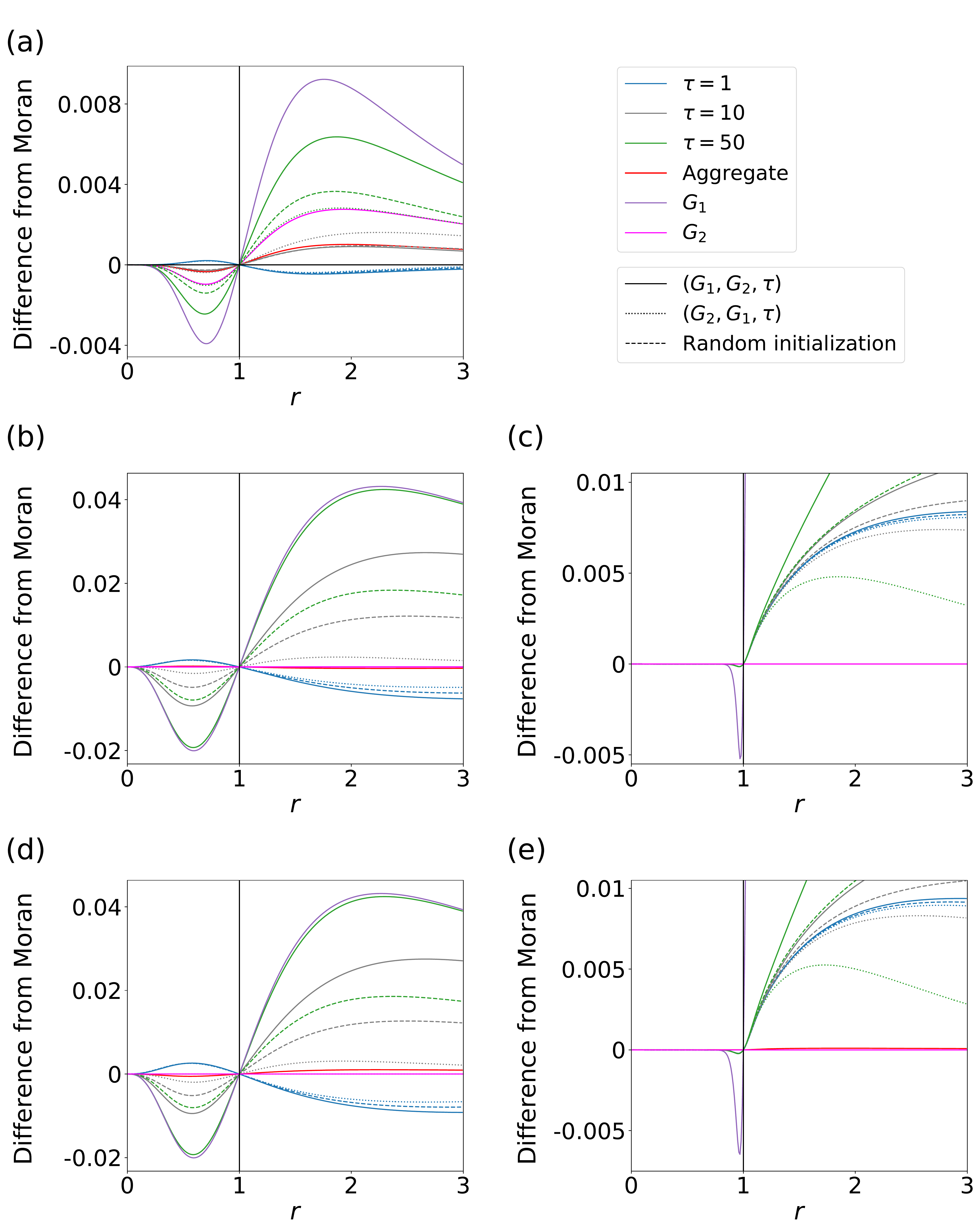}
			\caption{Fixation probability for switching networks \del{in which $G_2$ is the star graph} \add{with the order of $G_1$ and $G_2$ being swapped and with random initialization time}. In each panel, we show the difference in the fixation probability from the case of the Moran process. \del{(a) $N=4$ and the complete graph as $G_1$. (b) $N=50$ and the complete graph as $G_1$.
				(c) $N_1 = N_2 = 2$ and the complete bipartite graph as $G_1$. (d) $N_1 = N_2 = 20$ and the complete bipartite graph as $G_1$. We do not show the result for $G_1$ because both the complete graph and the complete bipartite graph are equivalent to the Moran process such that the difference from the Moran process is $0$ for any $r$.}
\add{\textbf{(a)} $G_1$ and $G_2$ given in Fig.~\ref{fig:ampsbecomesup}(a). In (b) and (c), $G_1$ is the star graph, and $G_2$ is the complete graph. (b) $N=4$. (c) $N=50$. In (d) and (e), $G_1$ is the star graph, and $G_2$ is the complete bipartite graph. (d) $N_1 = N_2 = 2$. (e) $N_1 = N_2 = 20$.
In (a), the results for all the three switching networks with $\tau=1$, shown by the blue lines, heavily overlap with each other, and those for $(G_1, G_2, 10)$, shown by the gray solid line, and those for the random initialization with $\tau=10$, shown by the gray dashed line, heavily overlap on top of each other. In (b)--(e), $G_2$ is equivalent to the Moran process. Therefore, the results for $G_2$, shown by the magenta solid line, completely overlap with the horizontal axis. Similarly, in (b) and (d), the results for $(G_2,G_1,50)$, shown by the green dotted lines, heavily overlap with the horizontal axis and are almost hidden behind the magenta solid lines. In addition, in (b), (c), and (e), the results for the aggregate network, shown by the red solid lines, almost completely overlap with or are very close to the horizontal axis.}
	}
			\label{fig:reversed_order_star_com_bipartite}
		\end{center}
	\end{figure}}
	%%%%%%%%%%%%%%%%%%%%%
	%%%%%%%%%%%%%%%%%%%%%

	\subsection*{B. Further examples of small  switching networks in which $G_1$ is the star graph}\label{sec:a2-further-examples-of-small-amplifying-switching-networks-in-which-g1-is-the-star-graph}
	
	In Fig.~\ref{fig:suppressors_additional}(a), we show the difference in the fixation probability from the case of the Moran process for the switching networks in which $G_1$ is the star graph and $G_2$ is the complete graph on $N=3$ nodes. We also plot the results for $G_1$, $G_2$, and the aggregate network. 
	It is known that $G_1$ is an amplifier \add{of selection} \cite{lieberman2005evolutionary} and that $G_2$ is equivalent to the Moran process.
	In contrast, the switching network with $\tau=1$ and the aggregate network are suppressors \add{of selection}.
	The aggregate network is much less suppressing than the switching network. The switching networks with $\tau \in \{10, 50\}$ are amplifiers \add{of selection}. 
	
	In Fig.~\ref{fig:suppressors_additional}(b), we show the results for the switching networks in which $G_1$ is the star graph and $G_2$ is the complete bipartite graph, $K_{\left(3,2\right)}$, on $N=5$ nodes. Note that both $G_1$ (i.e., star) \cite{lieberman2005evolutionary} and $G_2$ (i.e., complete bipartite graph $K_{\left(3,2\right)}$)~\cite{monk2014martingales} are amplifiers \add{of selection}. In contrast, as in Fig.~\ref{fig:suppressors_additional}(a), the switching network with $\tau=1$ (but not with $\tau \in \{10, 50 \}$) and the aggregate network are suppressors \add{of selection}, and the aggregate network is only weakly suppressing.
	
	%%%%%%%%%%%%%%%%%%%%%%%
	%%%%%%%%%%%%%%%%%%%%%%%
	
	\begin{figure}[t] % [!htbp]
		\begin{center}
			\includegraphics[width=13cm]{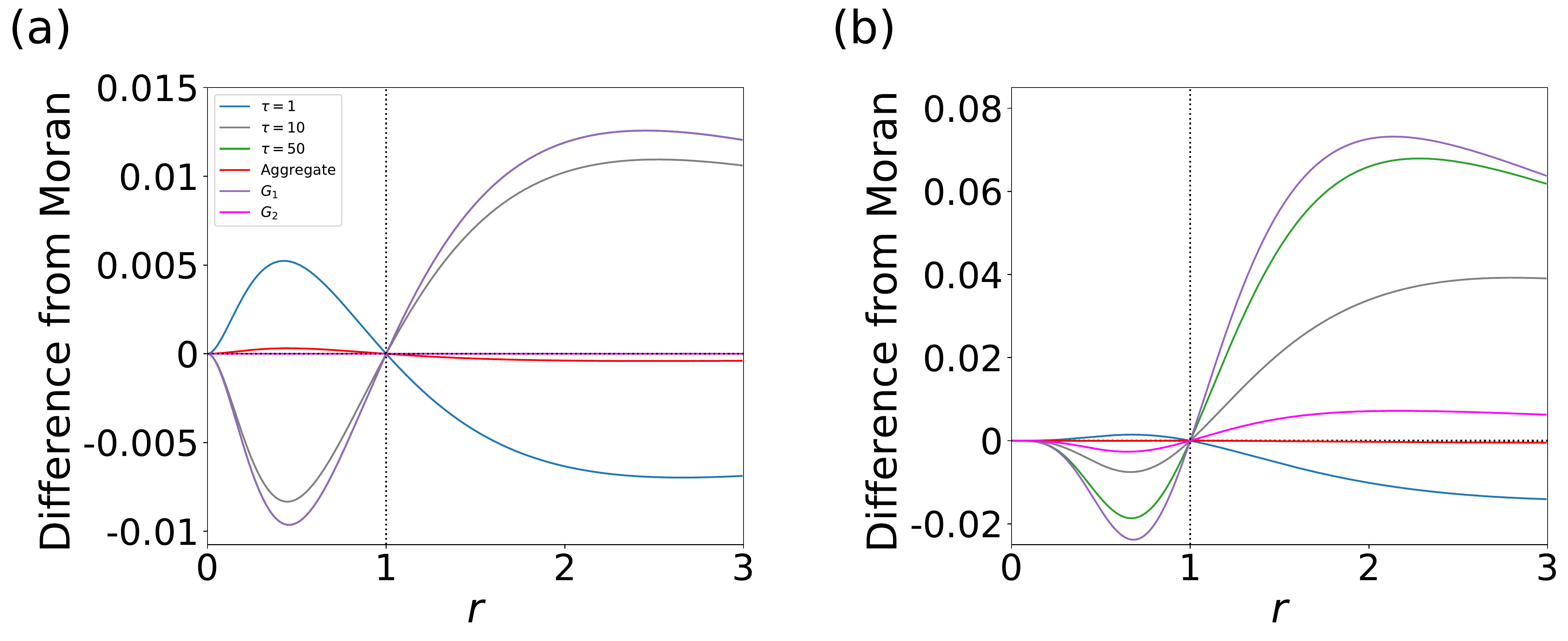}
			\caption{Fixation probability as a function of $r$ for two small switching networks. (a) Switching network with $N=3$ in which $G_1$ is the star graph and $G_2$ is the complete graph. (b) Switching network with $N=5$ in which $G_1$ is the star graph and $G_2$ is the complete bipartite graph $K_{\left(3,2\right)}$. \add{In (a), because $G_2$ is the complete graph, its plot is exactly on the horizontal axis.}}
			\label{fig:suppressors_additional}
		\end{center}
	\end{figure}
	%%%%%%%%%%%%%%%%%%%%%%%
	\subsection*{\add{C. Transition probabilities for the combination of the star graph and the complete bipartite graph}}\label{sec:adda3-transition-probabilities-for-the-combination-of-the-star-graph-and-the-complete-bipartite-graph}
		\add{The transition probability matrix for the star graph is given by}
\add{\begin{linenomath}\begin{equation}
		T^{\left(1\right)}_{\left(i,j,k\right)\rightarrow\left(i',j',k'\right)} = \begin{cases}
			\frac{r(j+k)}{C_3} & \text{if } i=0 \text{ and } i'=1,
			\\[1mm]
			\frac{N-1-j-k}{C_4} & \text{if } i = 1 \text{ and } i'=0,
			\\[1mm]
			\frac{1}{C_3}\cdot\frac{j}{N-1} & \text{if } i=0 \text{ and }j'= j-1,
			\\[1mm]
			\frac{1}{C_3}\cdot\frac{k}{N-1} & \text{if }i=0 \text{ and }k' = k-1,
			\\[1mm]
			\frac{r}{C_4}\cdot\frac{N_{1}-1-j}{N-1} & \text{if }i=1 \text{ and } j' = j+1,
			\\[1mm]
			\frac{r}{C_4}\cdot\frac{N_2-j}{N-1} & \text{if }i=1 \text{ and } k' = k+1,
			\\[1mm]
			1-\sum\limits_{\mathclap{\substack{(i'',j'',k'')\neq \\ (i,j,k)}}}T^{\left(1\right)}_{(i,j,k) \rightarrow (i'',j'',k'')} & \text{if } (i',j',k')=(i, j, k),
			\\[1mm]
			0 & \text{ otherwise,}
		\end{cases}
		\label{eq:star-bipartite-T1}
	\end{equation}\end{linenomath}}
\add{where}\begin{linenomath}
\begin{equation}
\add{C_3 = r(j+k)+(N-j-k)}
\end{equation}\end{linenomath}
and
\begin{linenomath}
\begin{equation}
\add{C_4 = r\left(j+k+1\right)+\left(N-j-k-1\right).}
\end{equation}\end{linenomath}
\add{The first line of Eq.~\eqref{eq:star-bipartite-T1} represents the probability that the type of the hub changes from the resident to mutant. For this event to occur, one of the $j+k$ leaf nodes occupied by the mutant must be chosen as parent, which occurs with probability $r (j+k) / C_3$. Then, because any leaf node is only adjacent to the hub node, the hub node is always selected for death. Therefore, the probability of $i$ changing from $0$ to $1$ is equal to $r(j+k) / C_3$. As another example, consider state $\left(1,j,k\right)$. For the state to change from $\left(1,j,k\right)$ to $ \left(1,j+1,k\right)$, the hub node, which the mutant type currently inhabits, must be selected as parent with probability $r/C_4$. Then, one of the $j$ leaf nodes of the resident type in $V_1$ must be selected for death, which occurs with probability $\left[(N_1- 1) - j\right]/\left(N-1\right)$. The fifth line of Eq.~\eqref{eq:star-bipartite-T1} is equal to the product of these two probabilities. One can similarly derive the other lines of Eq.~\eqref{eq:star-bipartite-T1}.}
	
	\add{The transition probability matrix for the complete bipartite graph is given by}
	\add{\begin{linenomath}\begin{equation}
		T^{\left(2\right)}_{\left(i,j,k\right)\rightarrow\left(i',j',k'\right)} = \begin{cases}
			\frac{r k}{C_3}\cdot\frac{1}{N_1} & \text{if }i=0 \text{ and } i'=1,
			\\[1mm]
			\frac{N_2-k}{C_4}\cdot\frac{1}{N_1} & \text{if }i=1 \text{ and } i'=0,
			\\[1mm]
			\frac{N_2-k}{C_3}\cdot\frac{j}{N_{1}} & \text{if } i=0 \text{ and } j' = j-1,
			\\[1mm]
			\frac{r k}{C_3}\cdot\frac{N_{1}-1-j}{N_{1}} & \text{if } i=0 \text{ and } j' = j +1,
			\\[1mm]
			\frac{N_{1}-j}{C_3}\cdot\frac{k}{N_{2}} & \text{if } i=0 \text{ and } k' = k-1,
			\\[1mm]
			\frac{rj}{C_3}\cdot\frac{N_{2}-k}{N_{2}}  & \text{if } i=0 \text{ and } k' = k+1,
			\\[1mm]
			\frac{N_{2}-k}{C_4}\cdot\frac{j}{N_{1}} & \text{if }i=1 \text{ and } j' = j-1,
			\\[1mm]
			\frac{rk}{C_4}\cdot\frac{N_{1}-1-j}{N_{1}} & \text{if }i=1 \text{ and } j' = j+1,
			\\[1mm]
			\frac{N_1-1-j}{C_4}\cdot \frac{k}{N_2} & \text{if }i=1 \text{ and } k' = k-1,
			\\[1mm]
			\frac{r\left(j+1\right)}{C_4}\cdot\frac{N_{2}-k}{N_2}  & \text{if }i=1 \text{ and } k'=k+1,
			\\[1mm]
			1-\sum\limits_{\mathclap{\substack{(i'',j'',k'')\neq \\ (i,j,k)}}}T^{\left(2\right)}_{\left(i,j,k\right)\rightarrow\left(i'',j'',k''\right)}  & \text{if } (i',j',k')=(i,j,k).
			\\[1mm]
			0 & \text{ otherwise.}
		\end{cases}
		\label{eq:star-bipartite-T2}		
	\end{equation}\end{linenomath}}
	\add{The first line of Eq.~\eqref{eq:star-bipartite-T2} represents the probability that the type of the hub changes from the resident to mutant. For this event to occur, one of the $k$ mutant nodes in $V_2$ must be selected as parent with probability $rk/C_3$. Then, the hub node must be selected for death with probability $1/N_1$ because each node in $V_2$ is only adjacent to all the $N_1$ nodes in $V_1$. Therefore, the probability of $i$ changing from $0$ to $1$ is equal to $\left(rk/
	C_3\right)\cdot \left(1/N_1\right)$.
	As another example, consider state $\left(1,j,k\right)$, in which there are $j+k+1$ mutants in total. For the state to change from $\left(1,j,k\right)$ to $ \left(1,j+1,k\right) $, 
	one of the $k$ mutant nodes in $V_2$ must first be selected as parent with probability $rk/C_4$. Then, one of the $j$ leaf nodes in $V_1$ of the resident type must be selected for death, which occurs with probability $(N_1-1-j)/N_1$. The eighth line of Eq.~\eqref{eq:star-bipartite-T2} is equal to the product of these two probabilities. One can similarly derive the other lines of Eq.~\eqref{eq:star-bipartite-T2}.}
	
	%%%%%%%%%%%%%%%%%%%%%%%%%%%%%%%%
	\subsection*{\add{D. Examples in which $G_1$ is the star graph, $G_2$ is the complete bipartite graph, and $N_1\neq N_2$}}\label{sec:addd-examples-in-which-g1-is-the-star-graph-g2-is-the-complete-bipartite-graph-and-n1neq-n2}
	
	\add{In this section we consider switching networks when $G_1$ is the star graph and $G_2$ is the complete bipartite graph. In Fig.~\ref{fig:unified_bipartite_plots}, we have shown the results for $N_1=N_2$. The complete bipartite graph $K_{\left(N_1,N_2\right)}$ is isothermal when $N_1=N_2$, whereas it is an amplifier \add{of selection} when $N_1\neq N_2$~\cite{monk2014martingales}. Therefore, the fixation probability for $(G_1, G_2, \tau)$ may be qualitatively different between the cases $N_1 = N_2$ and $N_1 \neq N_2$. In this section, we examine two switching networks when $N_1\neq N_2$.}
	
	\add{In Fig.~\ref{fig:star_bipartite_unsymmetric}(a), we show the difference in the fixation probability from the case of the Moran process for the switching networks in which $G_1$ is the star graph on $N=6$ nodes and $G_2$ is the complete bipartite graph $K_{\left(4,2\right)}$. We also plot the fixation probability for $G_1$, $G_2$, and the aggregate network. \del{As stated in the previous sections,} \add{Although} $G_1$ and $G_2$ are both amplifiers \add{of selection} \cite{lieberman2005evolutionary,monk2014martingales}, \del{.
	In contrast} \del{the switching network with $\tau=1$} \add{$(G_1, G_2, 1)$} is neither an amplifier nor a suppressor \add{of selection}. It transitions approximately at $r=2.062$ as $r$ increases, similar to static networks analyzed before \cite{alcalde2018evolutionary}.
%
% \jb{r=2.06222687 is the precise value.} \nm{OK.}
%
\add{Specifically, $(G_1, G_2, 1)$ is amplifying} when $r$ is approximately smaller than $2.062$ and suppressing when $r$ is approximately greater than $2.062$.  The switching networks with $\tau \in \{10, 50\}$ as well as the aggregate network are amplifiers of selection.}
	
	\add{In Fig.~\ref{fig:star_bipartite_unsymmetric}(b), we show the results for $(G_1, G_2, \tau)$ in which $G_1$ is the star graph on $N=40$ nodes and $G_2$ is the complete bipartite graph $K_{\left(30,10\right)}$. We note that both $G_1$ and $G_2$ are amplifiers of selection. In this case, all the switching networks and the static networks are amplifiers of selection, which is qualitatively the same result as that for $N_1 = N_2 = 20$ (see Figs.~\ref{fig:unified_bipartite_plots}(c) and (d)).}
	
	%%%%%%%%%%%%%%%%%%%%%%%
	%%%%%%%%%%%%%%%%%%%%%%%
	%\setcounter{figure}{2}
	\add{\begin{figure}[ht] % [!htbp]
		\begin{center}
			\includegraphics[width=13cm]{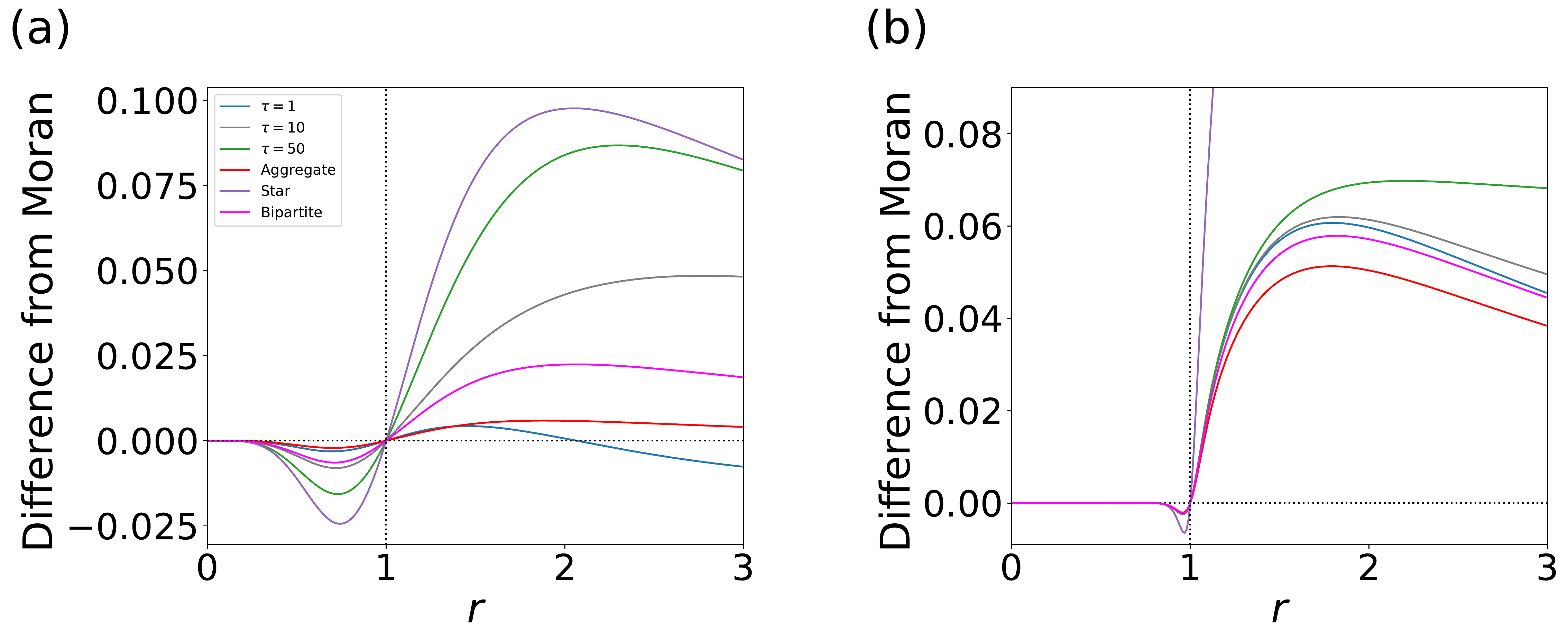}
			\caption{\add{Fixation probability as a function of $r$ when $G_1$ is the star graph, $G_2$ is the complete graph, and $N_1 \neq N_2$. (a) $\left(N_1,N_2\right)=\left(4,2\right)$. (b) $\left(N_1,N_2\right)=$$\left(30,10\right)$.}}
			\label{fig:star_bipartite_unsymmetric}
		\end{center}
	\end{figure}}

	%%%%%%%%%%%%%%%%%%%%%%%
	%%%%%%%%%%%%%%%%%%%%%%%%%%%%%%%%%
	
	%\printbibliography 
	%\bibliographystyle{unsrt}
	%\bibliography{References}

\end{document}